\documentclass[twocolumn, aps, pra, showpacs, superscriptaddress, floatfix, nofootinbib, 10pt]{revtex4-2}
\usepackage[colorlinks=true, allcolors=blue]{hyperref}
\usepackage[utf8]{inputenc}
\usepackage{physics,parskip}
\usepackage{graphicx}
\usepackage{amsmath,amssymb,amsthm}
\usepackage{xcolor}
\usepackage{bbold}
\usepackage{enumerate}
\usepackage{enumitem,tasks}
\usepackage{microtype}
\usepackage{tikz}
\usetikzlibrary{decorations.markings,quantikz}
\usetikzlibrary{fit,shapes.arrows}

\settasks{ 
    label = \roman*.,
    label-format = {\itshape}
}

\newcommand{\GHZ}{\mathrm{GHZ}}

\newcommand{\GME}{\mathrm{GME}}
\newcommand{\GMNL}{\mathrm{GMNL}}

\usepackage[normalem]{ulem}

\begin{abstract}
Bell's theorem proves that some quantum state correlations can only be explained by bipartite nonclassical resources.
The notion of genuinely multipartite nonlocality (GMNL) was later introduced to conceptualize the fact that nonclassical resources involving more than two parties in a nontrivial way may be needed to account for some quantum correlations. In this letter, we first recall the contradictions inherent to the historical definition of $\GMNL$.
Second, we turn to one of its redefinitions, called Local-Operations-and-Shared-Randomness $\GMNL$ (LOSR-GMNL), proving that all caterpillar graph states (including cluster states) have this second property.
Finally, we conceptualize a third, alternative definition, which we call Local-Operations-and-Neighbour-Communication $\GMNL$ (LONC-GMNL), that is adapted to situations in which short-range communication between some parties might occur. We show that cluster states do not have this third property, while GHZ states do.

Beyond its technical content, our letter illustrates that rigorous conceptual work is needed before applying the concepts of genuinely multipartite nonlocality, genuine multipartite entanglement or entanglement depth to benchmark the nonclassicality of some experimentally-produced quantum system. 
We note that most experimental works still use witnesses based on the historical definitions of these notions, which fail to reject models based on bipartite resources.

\end{abstract}

\begin{document}
\title{The genuinely multipartite nonlocality of graph states is model-dependent}
\author{Xavier Coiteux-Roy}
\thanks{These two authors contributed equally.}
\affiliation{School of Computation, Information and Technology, Technical University of Munich \& Munich Center for Quantum Science and Technology (MCQST), Munich, Germany.}  

\author{Owidiusz Makuta}
\thanks{These two authors contributed equally.}
\affiliation{Center for Theoretical Physics, Polish Academy of Sciences, Aleja Lotnik\'{o}w 32/46, 02-668 Warsaw, Poland.}

\author{Fionnuala Curran}
\affiliation{ICFO-Institut de Ciencies Fotoniques, The Barcelona Institute of Science and Technology, Castelldefels (Barcelona), Spain.}

\author{Remigiusz Augusiak}
\affiliation{Center for Theoretical Physics, Polish Academy of Sciences, Aleja Lotnik\'{o}w 32/46, 02-668 Warsaw, Poland.}

\author{Marc-Olivier Renou}
\affiliation{Inria Saclay, Bâtiment Alan Turing, 1, rue Honoré d’Estienne d’Orves 91120 Palaiseau.}
    \affiliation{CPHT, Ecole polytechnique, Institut Polytechnique de Paris, Route de Saclay – 91128 Palaiseau.}

\date{\today}
\newtheorem{prop}{Proposition}
\newtheorem{lemma}{Lemma}
\newtheorem{theorem}{Theorem}
\newtheorem{corollary}{Corollary}
\newtheorem{definition}{Definition}
\maketitle

\begin{figure*}
    \centering
 \centerline{
 {\hspace{-30pt}
    \resizebox{1.1\textwidth/2}{!}{\input{Figure1a.tex}}
    }\hspace{-20pt}
 {
    \resizebox{11\textwidth/54}{!}{\definecolor{myred}{RGB}{204,51,17}
\definecolor{myblue}{RGB}{0,119,187}
\definecolor{mygrey}{RGB}{135,135,135}
\definecolor{myteal}{RGB}{80,175,148}
\definecolor{myorange}{RGB}{238,119,51}
\definecolor{mymagenta}{RGB}{238,51,110}

\tikzset{player/.style={circle,draw=black,inner sep=0, minimum size=20pt,fill=white,thick
 }}

\tikzset{sourcea/.style={circle, fill=myred,inner sep=10/3pt,outer sep=0pt}}
\tikzset{sourceb/.style={circle, fill=myblue,inner sep=10/3pt,outer sep=0pt}}
\tikzset{sourcec/.style={circle, fill=myorange,inner sep=10/3pt,outer sep=0pt}}
\tikzset{sourced/.style={circle, fill=myteal,inner sep=10/3pt,outer sep=0pt}}

\tikzset{ligne/.style={very thick
}}

\centering
\begin{tikzpicture}[scale=0.2]
\centering
\usetikzlibrary{decorations.pathreplacing}
\node[player] (A0) at (5,8.66) {$A$}; %
\node[player] (B0) at (0,0) {$B$};
\node[player] (C0) at (10,0) {$C$};

\node[sourcec] (SA0) at (2.5,4.33) { }; %
\node[sourceb] (SB0) at (7.5,4.33) { };%
\node[sourced] (SC0) at (5,0) { };%

\draw[ligne,to-] (A0) -- (SA0);
\draw[ligne,to-] (B0) -- (SA0);
\draw[ligne,to-]  (A0) -- (SB0);
\draw[ligne,to-] (C0) -- (SB0);
\draw[ligne,to-]  (B0) -- (SC0);
\draw[ligne,to-] (C0) -- (SC0);

\begin{scope}[on background layer]
\node[inner sep=1pt] (L) at (5,3.5) {\large $\textcolor{mygrey}{{\lambda}}$};
\draw[ligne,-to,mygrey] (L) -- (A0);
\draw[ligne,-to,mygrey] (L) -- (B0);
\draw[ligne,-to,mygrey] (L) -- (C0);
\end{scope}

\node[anchor=center,inner sep=2pt] (inputA) at (2.5,8.66+3) {$\textcolor{myred}{x}$};
\node[anchor=center,inner sep=2pt] (outputA) at (5+2.5,8.66+3) {$\textcolor{myred}{a}$};
\node[anchor=center,inner sep=1pt] (inputB) at (0-2.5,-3) {$\textcolor{myred}{y}$};
\node[anchor=center,inner sep=2pt] (outputB) at (0+2.5,-3) {$\textcolor{myred}{b}$};
\node[anchor=center,inner sep=2pt] (inputC) at (10-2.5,0-3) {$\textcolor{myred}{z}$};
\node[anchor=center,inner sep=2pt] (outputC) at (10+2.5,0-3) {$\textcolor{myred}{c}$};

\draw[thick,-to,myred] (inputA) to[bend right] (A0);
\draw[thick,to-,myred] (outputA) to[bend left] (A0);
\draw[thick,-to,myred] (inputB) to[bend left] (B0);
\draw[thick,to-,myred] (outputB) to[bend right] (B0);
\draw[thick,-to,myred] (inputC) to[bend left] (C0);
\draw[thick,to-,myred] (outputC) to[bend right] (C0);

\node[align=left,anchor=north] at (5,-4.5) {b. Causal model of the\\ LOSR redefinition};
\end{tikzpicture}}
    }\hspace{-20pt}
    {
    \resizebox{2.1\textwidth/6}{!}{\definecolor{myred}{RGB}{204,51,17}
\definecolor{myblue}{RGB}{0,119,187}
\definecolor{mygrey}{RGB}{135,135,135}
\definecolor{myteal}{RGB}{80,175,148}
\definecolor{myorange}{RGB}{238,119,51}
\definecolor{mymagenta}{RGB}{238,51,110}

\tikzset{player/.style={circle,draw=black,inner sep=0, minimum size=20pt,fill=white,thick
 }}

\tikzset{sourcea/.style={circle, fill=myred,inner sep=10/3pt,outer sep=0pt}}
\tikzset{sourceb/.style={circle, fill=myblue,inner sep=10/3pt,outer sep=0pt}}
\tikzset{sourcec/.style={circle, fill=myorange,inner sep=10/3pt,outer sep=0pt}}
\tikzset{sourced/.style={circle, fill=myteal,inner sep=10/3pt,outer sep=0pt}}

\tikzset{ligne/.style={very thick
}}

\centering
\begin{tikzpicture}[scale=0.2]
\centering

\usetikzlibrary{decorations.pathreplacing}
\tikzset{ligne/.style={very thick,-to
}}

\node[player] (A0) at (0,0) {$A_{t_0}$}; 
\node[player] (B0) at (7.5,0) {$B_{t_0}$};
\node[player] (C0) at (15,0) {$C_{t_0}$};
\node[player] (D0) at (22.5,0) {$D_{t_0}$};
\node[] (E0) at (30,0) {$\Large \mathbf{\dots}$};

\node[player] (A1) at (0,-6) {$A_{t_1}$}; 
\node[player] (B1) at (7.5,-6) {$B_{t_1}$};
\node[player] (C1) at (15,-6) {$C_{t_1}$};
\node[player] (D1) at (22.5,-6) {$D_{t_1}$};
\node[] (E1) at (30,-6) {$\Large \mathbf{\dots}$};

\node[player] (A2) at (0,-12) {$A_{t_2}$}; 
\node[player] (B2) at (7.5,-12) {$B_{t_2}$};
\node[player] (C2) at (15,-12) {$C_{t_2}$};
\node[player] (D2) at (22.5,-12) {$D_{t_2}$};
\node[] (E2) at (30,-12) {$\Large \mathbf{\dots}$};

\draw[ligne] (A0) -- (A1);
\draw[ligne] (A0) -- (B1);
\draw[ligne] (B0) -- (B1);
\draw[ligne] (B0) -- (C1);
\draw[ligne] (C0) -- (C1);
\draw[ligne] (C0) -- (D1);
\draw[ligne] (D0) -- (D1);
\draw[ligne] (D0) -- (E1);

\draw[ligne] (A1) -- (A2);
\draw[ligne] (A1) -- (B2);
\draw[ligne] (B1) -- (B2);
\draw[ligne] (B1) -- (C2);
\draw[ligne] (C1) -- (C2);
\draw[ligne] (C1) -- (D2);
\draw[ligne] (D1) -- (D2);
\draw[ligne] (D1) -- (E2);

\begin{scope}[on background layer]
\node[inner sep=1pt] (L) at (17.5,3.5) {\large $\textcolor{mygrey}{{\lambda}}$};
\draw[ligne,-to,mygrey] (L) -- (A0);
\draw[ligne,-to,mygrey] (L) -- (B0);
\draw[ligne,-to,mygrey] (L) -- (C0);
\draw[ligne,-to,mygrey] (L) -- (D0);
\draw[ligne,-to,mygrey] (L) -- (E0);
\end{scope}

\node[inner sep=2pt] (inputA) at (0,4.5) {$\textcolor{myred}{x}$};
\node[inner sep=2pt] (inputB) at (7.5,4.5) {$\textcolor{myred}{y}$};
\node[inner sep=2pt] (inputC) at (15,4.5) {$\textcolor{myred}{z}$};
\node[inner sep=2pt] (inputD) at (22.5,4.5) {$\textcolor{myred}{w}$};
\node[inner sep=2pt] (outputA) at (0,-16.5) {$\textcolor{myred}{a}$};
\node[inner sep=2pt] (outputB) at (7.5,-16.5) {$\textcolor{myred}{b}$};
\node[inner sep=2pt] (outputC) at (15,-16.5) {$\textcolor{myred}{c}$};
\node[inner sep=2pt] (outputD) at (22.5,-16.5) {$\textcolor{myred}{d}$};

\draw[thick,-to,myred] (inputA) -- (A0);
\draw[thick,-to,myred] (inputB) -- (B0);
\draw[thick,-to,myred] (inputC) -- (C0);
\draw[thick,-to,myred] (inputD) -- (D0);
\draw[thick,to-,myred] (outputA) -- (A2);
\draw[thick,to-,myred] (outputB) -- (B2);
\draw[thick,to-,myred] (outputC) -- (C2);
\draw[thick,to-,myred] (outputD) -- (D2);

\node[align=left,anchor=north] at (15,-18) {c. Novel LONC explanatory causal model,\\ inspired by distributed computing};
\end{tikzpicture}}
    }}
 \caption{
    Three different formulations of the notion of genuine multipartite nonlocality. \\
    a. Svetlichny's LOCC causal explanatory model defines as $\GMNL$ the correlations $\vec{P}=\{p(a,b,c|x,y,z)\}$ that cannot be obtained in a causal structure in which one stochastically chosen party (selected by the die) is independent from the other two parties (e.g. if the die rolls $\mu_4$, then $B$ is independent from $A,C$). Svetlichny's $\GMNL$ correlations can thus be written as $\vec{P}\neq\int \textrm{d}\lambda\: \vec{Q}_{AB}^\lambda\vec{Q}_{C}^\lambda + \int \textrm{d}\mu\: \vec{R}_{BC}^\mu\vec{R}_{A}^\mu + \int \textrm{d}\nu\: \vec{S}_{AC}^\nu\vec{S}_{B}^\nu
$, where e.g. $\vec{Q}_{AB}^\lambda=\{q^\lambda(a,b|x,y)\}_{\lambda}$ may be a signalling distribution.\\
b. In the LOSR redefinition~\cite{PhysRevLett.127.200401},
$\vec{P}=\{p(a,b,c|x,y,z)\}$ is $\GMNL$ if it is not compatible with the represented causal model, where the principles of causality (also called No-Signalling and Independence) and of device-replication are assumed (see the proof of Theorem~\ref{theo: 4-cluster} and~Fig.~\ref{fig:intuitionc4}). Here, the shared randomness $\lambda$ is tripartite and the three bipartite sources are arbitrary no-signalling resources.\\
c. In this letter, we introduce the Local-Operations-and-Neighbour-Communication (LONC) framework, in which we define the LONC-$\GMNL$ correlations for a given underlying graph. Here we illustrate the correlations $\vec{P}=\{p(a,b,c,d,\dots|x,y,z,w,\dots)\}$ that can be generated in two rounds of one-way synchronous communication on a path (i.e. the correlations are LONC-$\GMNL_2$ for the directed path).}
    \label{fig:CausalExplanatoryModels}
\end{figure*}

\emph{Introduction---} 
Bell’s 1964 theorem~\cite{bell1964einstein} proved that two space-like separated parties can produce nonlocal correlations by performing local measurements on an appropriately entangled quantum system. This discovery of the nonlocality of quantum correlations has had a profound and lasting impact. 
At a fundamental level, quantum nonlocal correlations resist explanations in terms of local-hidden-variable models: they can score higher at Bell games (e.g. the CHSH game) than classical correlations~\cite{CHSHoriginal,mermin1990quantum,coiteux2019rgb}.
In terms of concrete applications, Bell nonlocality allows for the certification of the properties or functionalities of quantum devices (e.g. entanglement of a state~\cite{mayers2004selftesting, Supic2020SelfTestingReview, Supic2023QuantNetwSelfTestAllEntStates} or a measurement \cite{Renou2018SelfTestingEntMeas, Bancal2018DICertificationBSM}, security of a quantum key distribution protocol~\cite{Acin2007DIQKD,Zhang2022DIQKDExperiment1,Nadlinger2022DIQKDExperiment2}) based solely on the objective operational correlations between experimental events.

In this \emph{device-independent} (DI) paradigm~\cite{scarani2019bellNonlocality}, the certification is obtained without making any assumptions on the measurement apparatus. Here, one does not even make any assumptions on the physical theory governing them (only assuming it satisfies causality~\cite{ArianoOPT}), allowing one to consider and test alternative descriptions of nature.
As a result, DI is seen as the most secure way to assess the progress of quantum technologies, e.g. to benchmark a quantum device without needing to trust its designer or the measurement apparatus used to characterize it.
This contrasts with the \emph{device-dependent} (DD) paradigm, where quantum theory is accepted (e.g. Tsirelson's $2\sqrt{2}$ bound for the CHSH game holds, while a score of up to 4 is possible in the DI paradigm~\cite{popescu1997PRbox})  and the detectors are trusted. Intermediary concepts of semi-device-independence have also been introduced~\cite{Pawwlowski2011SemiDI,VanHimbeeck2017SemiDI}.

The ability to manipulate large nonclassical quantum systems is seen as a key resource in many applications of quantum theory.
The desire to certify this ability in a DI way has motivated the concept of 
\emph{nonlocality of depth $n$} (which we abbreviate as $\GMNL_{n}$). Systems that are $\GMNL_{n}$ are systems that produce correlations whose nonlocality cannot be understood as being `obtained by combining many states composed of at most $(n-1)$-partite constituents'. Furthermore, when a system of $n$ parties exhibits $\GMNL_{n}$ correlations, we say the system is (maximally) \textit{genuinely multipartite nonlocal} (GMNL). A typical example of a state that should intuitively produce $\GMNL_3$ (hence $\GMNL$) correlations is the three-party $\ket{\GHZ_3}=(\ket{000}+\ket{111})/\sqrt{2}$ state, while, by contrast, the state $\ket{\phi^ +}\otimes\ket{0}$, which is composed of an EPR pair and an (independent) ancillary qubit, should intuitively produce only $\GMNL_2$ correlations.

The task of extending this preliminary idea into a general classification of all states, that is, proposing a universal mathematical definition for the concept of $\GMNL$, is not trivial.
It is, however, of critical importance, as many experimental works assess the large nonclassicality of their experimentally-produced systems based on the notion of GMNL~\cite{Riedel2010GMEAtomChip,Lucke2014GMEDickeState, Cao2023GMESupercondQubits,Bornet2023GMESpinSqueezing}.
An appropriate definition requires first the introduction of a concrete physical model expressing what is meant by `obtained by combining many states composed of at most $(n-1)$-partite constituents,' which we call the \emph{causal explanatory model}. 
Then, by definition, a state producing correlations that cannot be explained in this causal explanatory model will be called $\GMNL_n$.
The first historical definition of $\GMNL$ was proposed by Svetlichny \cite{PhysRevD.35.3066} (see Fig.~\ref{fig:CausalExplanatoryModels}a for a detailed description) in the context of Local Operations and Classical Communication (LOCC). 
It characterizes $\GMNL$ as the direct DI counterpart to the DD definition by Seevinck and Uffink of \emph{genuine multipartite entanglement}~\cite{Seevinck2001GME}. 
As we detail later, the LOCC definition by Svetlichny, as well as the one by Seevinck and Uffink, are known (since at least a decade ago~\cite{Gallego2012GMNLisBAD}) to be paradoxical when used to assess the large nonclassical nature of practical quantum systems. 
This has motivated the redefinition of these concepts in the framework of Local-Operations-and-Shared-Randomness (LOSR), which resulted in a total redefinition of the underlying causal explanatory model to be excluded~\cite{Navascu_s_2020,PhysRevLett.127.200401}, inspired by the concept of quantum network nonlocality~\cite{Tavakoli2022ReviewNonlocality, BeyondBellII, Renou2022NonlocForGenericNetworks,Renou2019GenuineQNonlocalityTriangle,Branciard2010BilocalInequality, Rosset2016BilocalInequalityAndMore}. The LOSR-$\GMNL$ definition is detailed in Fig.~\ref{fig:CausalExplanatoryModels}b.

Very few states are known to produce $\GMNL$ correlations under the LOSR framework.
In \cite{Navascu_s_2020} and \cite{PhysRevLett.127.200401} it was proven that $\ket{\GHZ}$ is $\GME$ and $\GMNL$ according to the redefinitions of Fig.~\ref{fig:CausalExplanatoryModels}b, which we denote by LOSR-$\GME$ and LOSR-$\GMNL$ respectively. The existence of LOSR-$\GMNL$ correlations has also been demonstrated experimentally~\cite{Cao2022ExperimentGMNL,Mao2022ExperimentGMNLJingyunFan}.
Recent results have shown that $n$-partite graph states, a well-studied family of multipartite quantum states~\cite{hein2004multiparty} which we recall below, are LOSR-$\GME_{3}$ \cite{Makuta2023-kj, wang2022quantum}.
Importantly, while appropriate in some situations, the new LOSR definition of Fig.~\ref{fig:CausalExplanatoryModels}b may not be suitable for all experiments. 
One could consider, for example, a
1$D$ condensed matter system, in which subsystems are so close to each other that short-range communication between adjacent sites cannot be ruled out. In such a scenario, an alternative causal explanatory model should be used which allows for communication between neighbours.

In this letter, we focus our analysis on the (more demanding) DI paradigm, discussing the implications for the (weaker) DD paradigm when applicable. We show that all cluster states, and in fact, all \emph{caterpillar graph states} (shown in Fig.~\ref{fig:graphstates}), are LOSR-$\GMNL$, in a noise-robust way. Our results improve significantly on previous techniques, which could only provide lower bounds of degree~3 on the nonlocality of those states, and were restricted to the weaker device-dependent setting of $\GME$~\cite{Makuta2023-kj, wang2022quantum}.
This is thus the main technical contribution of our letter. 
Second, motivated by experimental considerations, we introduce a new framework allowing Local Operations and Neighbour Communication (LONC). 
Its associated causal explanatory model (see Fig.~\ref{fig:CausalExplanatoryModels}c) follows $t$ communication steps along a given network. 
This leads us to a new notion of certifiable multipartite nonlocality which we call LONC-$\GMNL$, according to which states can be classified quantitatively as LONC-$\GMNL_t$ for a given network. We prove that, in this new LONC-$\GMNL$ definition, $\ket{\GHZ}$ and cluster states belong respectively to opposite ends of the complexity spectrum: for the directed path communication graph, $\ket{\GHZ}$ is maximally LONC-$\GMNL$ while the cluster state is only LONC-$\GMNL_2$ (this finding improves a result of~\cite{le2019quantum}). Characterizing this distinction by introducing the LONC framework is the main conceptual contribution of our letter.

\emph{LOCC and LOSR lead to two different notions of $\GMNL$---} 
The historical Svetlichny definition of $\GMNL$ is anchored in the LOCC framework and is based on the explanatory model of Fig.~\ref{fig:CausalExplanatoryModels}a. 
There are two important issues with this model.
First, multiple copies of few-partite states distributed among different parties can simulate many-partite $\GMNL$. For example, the composition $\ket{\phi^ +}_{AB_1}\otimes\ket{\phi^ +}_{B_2C}$ of two bipartite EPR states, where $A,C$ respectively own one qubit and $B$ owns two qubits, is Svetlichny-$\GMNL_3$ despite comprising only bipartite sources (more generally, $n$ copies of $\ket{\phi^ +}$ can be Svetlichny-$\GMNL_{n+1}$). 
This is highly problematic in a device-independent context where the large nonclassicallity of a quantum source should be benchmarked without trusting its designer or the measurement apparatus. Indeed, a designer only able to produce bipartite $\ket{\phi^ +}$ states could claim to produce $\GMNL_{n}$ states with arbitrarily large $n$~\footnote{Note that while the bipartite structure of $\ket{\phi^ +}_{AB_1}\otimes\ket{\phi^ +}_{B_2C}$ is apparent, $B$ could obfuscate that structure by encoding their state in a four-level system and applying some local unitary.}.
Second, more conceptually, this definition involves signalling distributions (see caption of Fig.~\ref{fig:CausalExplanatoryModels}b), which is not physically motivated. 
While the latter shortcoming was addressed in~\cite{bancal2013definitions}, patching the first issue required a shift of the concept of $\GMNL$ from the LOCC framework towards the LOSR framework.
This led to the introduction of the causal explanatory model which we detail in Fig.~\ref{fig:CausalExplanatoryModels}b. 

\emph{Graph states and cluster states---} 
Graph states are a family of states corresponding to graphs where vertices represent subsystems and edges determine the entanglement between these subsystems (see Fig.~\ref{fig:graphstates}).
Due to their inherent symmetries, as well as due to the simplicity of applying Clifford gates on graph states, these states have found numerous applications in quantum computation~\cite{PhysRevLett.86.5188, Briegel2009-wb, doi:10.1146/annurev-conmatphys-020911-125041, NIELSEN2006147}, quantum error correction~\cite{Yao2012-du, PhysRevA.52.R2493},  quantum cryptography~\cite{markham2008graph,PhysRevA.59.1829}, and even quantum metrology~\cite{Toth_2014, nguyen2023harnessing}. Furthermore, their high degree of nonclassicality has made them a frequent subject of entanglement~\cite{hein2006entanglement,PhysRevA.69.062311} and nonlocality~\cite{PhysRevA.71.042325,PhysRevLett.95.120405,PhysRevLett.124.020402} studies.

In our main text, we focus for simplicity on the 4-cluster state $\ket{C_4}$, which is a graph state corresponding to four vertices on a line. For a particular choice of local basis, it can be written as
    \begin{equation}
       \ket{C_4} := \frac{1}{\sqrt{2}}\Big(\ket{00}\ket{\phi^+}+\ket{11}\ket{\phi^-}\Big) 
       \,,
    \end{equation}
where $\ket{\phi^\pm}:=\frac{1}{\sqrt{2}}\big(\ket{00}\pm\ket{11}\big)$. 
For a more detailed introduction to graph states, see Appendix~\ref{app:graph_states}.

\begin{figure}
    \centering
   \definecolor{myred}{RGB}{204,51,17}
\definecolor{myblue}{RGB}{0,119,187}
\definecolor{mygrey}{RGB}{187,187,187}
\definecolor{myteal}{RGB}{0,153,136}
\definecolor{myorange}{RGB}{238,119,51}
\definecolor{mymagenta}{RGB}{238,51,110}

\resizebox{\columnwidth}{!}{%
\begin{tikzpicture}[scale=0.95]
\centering

\tikzset{vertex/.style={circle,inner sep=5pt,outer sep=-1/2pt, very thick,fill=black!90}}

\tikzset{svertex/.style={circle,inner sep=5pt,outer sep=-1/2pt,very thick,fill=myblue!100}}
\tikzset{cutvertex/.style={circle,inner sep=5pt,outer sep=-1/2pt,very thick,fill=mygrey}}

\tikzset{ligne/.style={very thick,black
}}

\tikzset{2ligne/.style={very thick, myblue}}
\tikzset{3ligne/.style={very thick, mygrey}}

\node at (-2,0) {$~$};
\node at (11,0) {$~$};

\node[text width=4.2cm, align=left,anchor=north] at (2,-1.75) {\large a.~A $4$-line graph};

\node[text width=4.8cm, align=left,anchor=north] at (8.5,-1.75) {\large b.~A caterpillar graph};

\node[vertex] (A0) at (0,0) {};
\node[vertex] (A1) at (1,0) {};
\node[vertex] (A2) at (2,0) {};
\node[vertex] (A3) at (3,0) {};
\begin{scope}[on background layer]
\draw[ligne] (A0) -- (A1) -- (A2) -- (A3);
\end{scope}
\begin{scope}[shift={(7,0)}]
\node[vertex] (A1) at (1,0) {};
\node[vertex] (A2) at (2,0) {};
\node[svertex] (A3) at (3,0) {};
\node[svertex] (B0) at (5/3,-1) {};
\node[svertex] (B1) at (7/3,-1) {};
\node[svertex] (B2) at (9/3,-1) {};
\node[svertex] (B3) at (7/3,1) {};
\node[svertex] (B5) at (5/3,1) {};
\node[svertex] (B6) at (3/3,-1) {};
\node[svertex] (C3) at (0,-1) {};
\node[svertex] (E) at (-2/3,1) {};
\node[vertex] (C4) at (-1,0) {};
\node[vertex] (D) at (0,0) {};

\begin{scope}[on background layer]
\draw[ligne] (A1)-- (A2);
\draw[2ligne] (E) -- (C4);
\draw[2ligne] (A3) -- (A2) -- (B0);
\draw[2ligne] (A2) -- (B1);
\draw[2ligne] (A2) -- (B2);
\draw[2ligne] (A1) -- (B5);
\draw[2ligne] (A2) -- (B3);
\draw[2ligne] (A1) -- (B6);
\draw[2ligne] (C4) -- (C3);
\draw[ligne] (D) -- (A1);
\draw[ligne] (D) -- (C4);

\end{scope}

\end{scope}

\end{tikzpicture}}
    \caption{a.~Graph states are defined from graphs. An $n$-cluster state corresponds to a line over $n$ vertices.\\
    b.~A caterpillar graph is a connected acyclic graph whose vertices of degree greater than one form a path (the ``spine'' in black above). All vertices of degree 1 (the ``legs'', in blue above) hence connect only to the spine.
   }\label{fig:graphstates}
\end{figure}
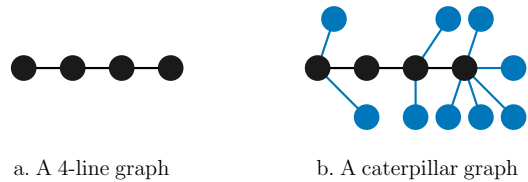

\emph{All caterpillar graph states are genuinely multipartite nonlocal---}
This section contains the main technical contribution of our letter, which is to prove that all caterpillar graph states are LOSR-$\GMNL$ in a noise-tolerant way. 
Here, we present the proof idea (for the detailed proof, see Appendix~\ref{app: 4-cluster}) using as an example the 4-partite cluster state $\ket{C_{4}}$, which is one of the simplest representatives of caterpillar graph states. The general proof for any caterpillar graph state is relegated to Appendix~\ref{app:C_N}.

\begin{theorem}\label{theo: 4-cluster}
All caterpillar graph states are LOSR-$\GMNL$ in a noise-robust way. \\
In particular, any state violating the following is $\GMNL_4$,\begin{equation}
2\langle A_0 B_0 \rangle + 2 \langle C_0 D_2 \rangle+ 2\langle A_1 B_1 D_2 \rangle + I^{BCD} \leqslant 8\,,\label{eqn: 4_cluster}
\end{equation}
where $I^{BCD}:=\langle C_0 D_{0} \rangle +\langle  C_0 D_{1} \rangle +  \langle B_0 C_1 D_{0} \rangle - \langle B_0 C_1 D_{1} \rangle $ is a distributed variant
of the CHSH inequality.

Under suitable measurements, the $\ket{C_4}$ cluster state violates this bound by  reaching $6+2\sqrt{2} > 8$.
\end{theorem}
\begin{proof}
The proof of Eq.~\eqref{eqn: 4_cluster} is based on the inflation technique~\cite{wolfe2019inflation,PhysRevLett.127.200401,PhysRevA.104.052207}. 
We assume that in the underlying causal model, the alternative explanation for the observed correlations is compatible with the principle of causality~\cite{ArianoOPT}, also called No-Signalling and Independence~\cite{Beigi2022CovDecompoNSI, gisin2020constraints}.  Its intuition is illustrated in Fig.~\ref{fig:intuitionc4}.
The quantum violation can be obtained by measuring on the state $\ket{C_4}_{ABCD}$ the following observables: $Z$ for $A_0, B_0, C_0, D_2$; $X$ for $A_1, B_1, C_1$; and, respectively, $(Z+ X) / \sqrt{2},(Z- X) / \sqrt{2}$ for $D_0, D_1$. 
\end{proof}

\definecolor{myred}{RGB}{204,51,17}
\definecolor{myblue}{RGB}{0,119,187}
\definecolor{mygrey}{RGB}{187,187,187}
\definecolor{myteal}{RGB}{80,175,148}
\definecolor{myorange}{RGB}{238,119,51}
\definecolor{mymagenta}{RGB}{238,51,110}

\tikzset{player/.style={circle,draw=black,inner sep=0, minimum size=20pt,fill=white,thick,font=\large
 }}

\tikzset{sourcea/.style={circle, fill=myred,inner sep=10.5/3pt,outer sep=-1pt}}
\tikzset{sourceb/.style={circle, fill=myblue,inner sep=10.5/3pt,outer sep=-1pt}}
\tikzset{sourcec/.style={circle, fill=myorange,inner sep=10.5/3pt,outer sep=-1pt}}
\tikzset{sourced/.style={circle, fill=myteal,inner sep=10.5/3pt,outer sep=-1pt}}

\tikzset{ligne/.style={ultra thick
}}

\tikzset{fleche/.style={-to,line width=1mm,color=black!50}}

\begin{figure}
\centering
\resizebox{\columnwidth}{!}{%
\begin{tikzpicture}[scale=0.18]
\centering

\node[player] (A0) at (0,10) {$A$};
\node[player] (B0) at (10,10) {$B$};
\node[player] (C0) at (0,0) {$C$};
\node[player] (D0) at (10,0) {$D$};

\node[sourcea] (SA0) at (10/3,20/3) { };
\node[sourceb] (SB0) at (20/3,20/3) { };
\node[sourcec] (SC0) at (10/3,10/3) { };
\node[sourced] (SD0) at (20/3,10/3) { };
\begin{scope}[on background layer]
\draw[ligne] (A0) -- (SA0) (A0) -- (SB0) (A0) -- (SC0);
\draw[ligne] (B0) -- (SA0) (B0) -- (SB0) (B0) -- (SD0);
\draw[ligne] (C0) -- (SA0) (C0) -- (SD0) (C0) -- (SC0);
\draw[ligne] (D0) -- (SD0) (D0) -- (SB0) (D0) -- (SC0);
\end{scope}
\node at (5,-4) {\Large $\mathcal{I}_0$};

\draw[fleche] (12,5) -- (16,5) ;
\draw[fleche] (32,-15) -- (36,-15) ;

\draw[fleche] (32,5) -- (36,5) ;
\draw[fleche] (12,-15) -- (16,-15) ;
\draw[fleche] (53,5) -- ({53+4*cos(35)},{5-4*sin(35}) ;
\draw[fleche] (53,-20+5) -- ({53+4*cos(35)},{-20+5+4*sin(35}) ;

\begin{scope}[shift={(20,0)}]
\node[player] (A0) at (0,10) {$A$};
\node[player] (B0) at (10,10) {$B$};
\node[player] (C0) at (0,0) {$C$};
\node[player] (D0) at (10,0) {$D$};

\node[sourcea] (SA0) at (10/3,20/3) { };
\node[sourceb] (SB0) at (20/3,20/3) { };
\node[sourcec] (SC0) at (10/3,10/3) { };
\node[sourcec] (SC1) at (0,5) { };
\node[sourced] (SD0) at (20/3,10/3) { };
\begin{scope}[on background layer]
\draw[ligne] (A0) -- (SA0) (A0) -- (SB0) (A0) -- (SC1);
\draw[ligne] (B0) -- (SA0) (B0) -- (SB0) (B0) -- (SD0);
\draw[ligne] (C0) -- (SA0) (C0) -- (SD0) (C0) -- (SC0);
\draw[ligne] (D0) -- (SD0) (D0) -- (SB0) (D0) -- (SC0);
\end{scope}
\node at (5,-4) {\Large $\mathcal{I}_1$};
\end{scope}

\begin{scope}[shift={(40,0)}]
\node[player] (A0) at (0,10) {$A$};
\node[player] (B0) at (10,10) {$B$};
\node[player] (C0) at (0,0) {$C$};
\node[player] (D0) at (10,0) {$D$};

\node[sourcea] (SA0) at (10/3,20/3) { };
\node[sourceb] (SB0) at (20/3,20/3) { };
\node[sourcec] (SC0) at (10/3,10/3) { };
\node[sourcec] (SC1) at (0,5) { };
\node[sourced] (SD0) at (20/3,10/3) { };
\node[sourced] (SD1) at (10,5) { };
\begin{scope}[on background layer]
\draw[ligne] (A0) -- (SA0) (A0) -- (SB0) (A0) -- (SC1);
\draw[ligne] (B0) -- (SA0) (B0) -- (SB0) (B0) -- (SD1);
\draw[ligne] (C0) -- (SA0) (C0) -- (SD0) (C0) -- (SC0);
\draw[ligne] (D0) -- (SD0) (D0) -- (SB0) (D0) -- (SC0);
\end{scope}
\node at (5,-4) {\Large $\mathcal{I}_2$};
\end{scope}

\begin{scope}[shift={(0,-10)},yscale=-1]
\node[player] (A0) at (0,10) {$A'$};
\node[player] (B0) at (10,10) {$B'$};
\node[player] (C0) at (0,0) {$C$};
\node[player] (D0) at (10,0) {$D$};

\node[sourcea] (SA0) at (10/3,20/3) { };
\node[sourceb] (SB0) at (20/3,20/3) { };
\node[sourcec] (SC0) at (10/3,10/3) { };
\node[sourced] (SD0) at (20/3,10/3) { };
\begin{scope}[on background layer]
\draw[ligne] (A0) -- (SA0) (A0) -- (SB0) (A0) -- (SC0);
\draw[ligne] (B0) -- (SA0) (B0) -- (SB0) (B0) -- (SD0);
\draw[ligne] (C0) -- (SA0) (C0) -- (SD0) (C0) -- (SC0);
\draw[ligne] (D0) -- (SD0) (D0) -- (SB0) (D0) -- (SC0);
\end{scope}
\node at (5,14) {\Large $\mathcal{I}_3$};
\end{scope}

\begin{scope}[shift={(20,-10)},yscale=-1]
\node[player] (A0) at (0,10) {$A'$};
\node[player] (B0) at (10,10) {$B'$};
\node[player] (C0) at (0,0) {$C$};
\node[player] (D0) at (10,0) {$D$};

\node[sourcea] (SA0) at (10/3,20/3) { };
\node[sourcea] (SA1) at (0,5) { };
\node[sourceb] (SB0) at (20/3,20/3) { };
\node[sourcec] (SC0) at (10/3,10/3) { };
\node[sourced] (SD0) at (20/3,10/3) { };
\begin{scope}[on background layer]
\draw[ligne] (A0) -- (SA0) (A0) -- (SB0) (A0) -- (SC0);
\draw[ligne] (B0) -- (SA0) (B0) -- (SB0) (B0) -- (SD0);
\draw[ligne] (C0) -- (SA1) (C0) -- (SD0) (C0) -- (SC0);
\draw[ligne] (D0) -- (SD0) (D0) -- (SB0) (D0) -- (SC0);
\end{scope}
\node at (5,14) {\Large $\mathcal{I}_4$};
\end{scope}

\begin{scope}[shift={(40,-10)},yscale=-1]
\node[player] (A0) at (0,10) {$A'$};
\node[player] (B0) at (10,10) {$B'$};
\node[player] (C0) at (0,0) {$C$};
\node[player] (D0) at (10,0) {$D$};

\node[sourcea] (SA0) at (10/3,20/3) { };
\node[sourcea] (SA1) at (0,5) { };
\node[sourceb] (SB0) at (20/3,20/3) { };
\node[sourceb] (SB1) at (10,5) { };
\node[sourcec] (SC0) at (10/3,10/3) { };
\node[sourced] (SD0) at (20/3,10/3) { };
\begin{scope}[on background layer]
\draw[ligne] (A0) -- (SA0) (A0) -- (SB0) (A0) -- (SC0);
\draw[ligne] (B0) -- (SA0) (B0) -- (SB0) (B0) -- (SD0);
\draw[ligne] (C0) -- (SA1) (C0) -- (SD0) (C0) -- (SC0);
\draw[ligne] (D0) -- (SD0) (D0) -- (SB1) (D0) -- (SC0);
\end{scope}
\node at (5,14) {\Large $\mathcal{I}_5$};
\end{scope}

\begin{scope}[shift={(180/3,-5)}]
\node[player] (A0) at (0,10) {$A$};
\node[player] (B0) at (10,10) {$B$};
\node[player] (C0) at (0,0) {$C$};
\node[player] (D0) at (10,0) {$D$};

\node[sourcea] (SA0) at (10/3,20/3) { };
\node[sourceb] (SB0) at (20/3,20/3) { };
\node[sourcec] (SC1) at (0,5) { };
\node[sourced] (SD1) at (10,5) { };
\begin{scope}[on background layer]
\draw[ligne] (A0) -- (SA0) (A0) -- (SB0) (A0) -- (SC1);
\draw[ligne] (B0) -- (SA0) (B0) -- (SB0) (B0) -- (SD1);
\draw[ligne] (C0) -- (SA0) ;
\draw[ligne] (D0) -- (SB0) ;
\end{scope}
\begin{scope}[yscale=-1]
\node[player] (A0) at (0,10) {$A'$};
\node[player] (B0) at (10,10) {$B'$};

\node[sourcea] (SA0) at (10/3,20/3) { };
\node[sourceb] (SB0) at (20/3,20/3) { };
\node[sourcec] (SC0) at (10/3,10/3) { };
\node[sourced] (SD0) at (20/3,10/3) { };
\begin{scope}[on background layer]
\draw[ligne] (A0) -- (SA0) (A0) -- (SB0) (A0) -- (SC0);
\draw[ligne] (B0) -- (SA0) (B0) -- (SB0) (B0) -- (SD0);
\draw[ligne] (C0) -- (SD0) (C0) -- (SC0);
\draw[ligne] (D0) -- (SD0) (D0) -- (SC0);
\end{scope}
\end{scope}

\node at (5,-14.5) {\Large $\mathcal{J}$};

\end{scope}

\end{tikzpicture}}
\caption{Intuition for the proof of Ineq.~\ref{eqn: 4_cluster} in Thm.~1. We assume here by contradiction that the quantum violation $\langle A_0B_0 \rangle = \langle C_0D_2 \rangle = \langle A_1B_1D_2 \rangle=1,\, I^{BCD}=2\sqrt{2}$ can be obtained without any 4-way nonclassical cause, that is, it can be obtained from four 3-way nonclassical causes as in the networks $\mathcal{I}_0, \mathcal{I}_3$, where the colored disks represent general tripartite nonsignalling sources (the inclusion of classical randomness shared between all parties would not change our proof). 
We show that this would imply that in the rightmost inflated scenario $\mathcal{J}$, $C$ would be part of the violation of a distributed variant of the CHSH inequality $I^{ACD}_{(\mathcal{J})}$ (the quantity $I^{ACD}$ is defined as $I^{BCD}$ but with Bob replacing Alice), but his output could be guessed by $A', B'$ together, which is impossible as it contradicts the monogamy of CHSH correlations.
We reach this conclusion by following two different lines of inflation reasonings: $\mathcal{I}_0 \rightarrow\mathcal{I}_1\rightarrow\mathcal{I}_2\rightarrow\mathcal{J}$ (top) and $\mathcal{I}_3 \rightarrow\mathcal{I}_4\rightarrow\mathcal{I}_5\rightarrow\mathcal{J}$ (bottom). 
We repeatedly use the argument that the behaviour of two subgroups of parties in two of these networks $\mathcal{I}_k, \mathcal{I}_l$ must be identical whenever their associated subnetworks are isomorphic. For instance, isolate $B$, $C$ and $D$ in $\mathcal{I}_0$ and $\mathcal{I}_1$. As their associated subnetworks (that is, the subnetworks in $\mathcal{I}_0, \mathcal{I}_1$ that are composed of $B, C, D$ and the sources they are connected to) are isomorphic, we must have $I^{BCD}_{(\mathcal{I}_0)}=I^{BCD}_{(\mathcal{I}_1)}$. Since we assumed that $I^{BCD}_{(\mathcal{I}_0)}=2\sqrt{2}$, we deduce that $I^{BCD}_{(\mathcal{I}_1)}=2\sqrt{2}$.
\\
By repeatedly using this type of argument along these networks, it is possible to show that $(i): \langle A_1'B_1'C_0 \rangle _{(\mathcal{J})}=1$ and $(ii): I^{ACD}_{(\mathcal{J})}=2\sqrt{2}$. $(i)$ implies that the product of the outputs of $A', B'$ on their input $1$ is equal to the output of $C$ on his input $0$. $(ii)$ shows that the output of $C$ on input $0$ is part of a violation of the distributed variant of the CHSH inequality $I^{ACD}_{(\mathcal{J})}\leq 2$.
This directly contradicts the monogamy of nonlocal correlations, which states that outputs violating CHSH cannot be guessed with high probability by a third party~\cite{Augusiak_2014}.
\label{fig:intuitionc4}}
\end{figure}
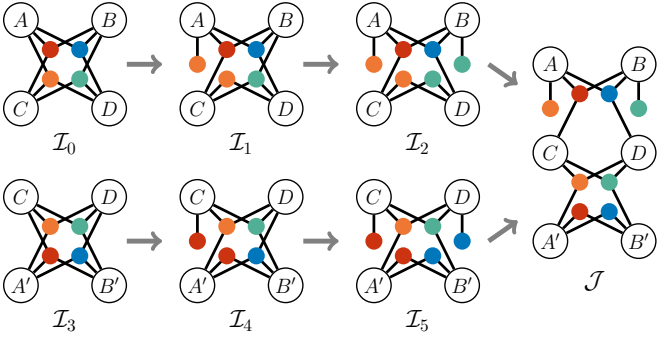

\emph{Higher dimensions---}
We show in Appendix~\ref{app:C_N^d} that cluster graph states of higher local dimensions (beyond qubits) are LOSR-$\GMNL$ for any prime local dimension and any number of subsystems, and in Appendix~\ref{app:ghz} we show that the GHZ state is LOSR-$\GMNL$ for any local dimension and any number of subsystems.

\emph{From LOSR to LONC---}
We now turn to the main conceptual contribution of our letter, which is to adapt the concept of $\GMNL$ to the Local-Operations-and-Neighbour-Communication (LONC) framework, which relaxes the no-signalling assumption of the LOSR framework. 
We recall that the concept of LOSR-$\GMNL$, shown in Fig.~\ref{fig:CausalExplanatoryModels}b, relies on the assumption that absolutely no communication occurs between any of the parties.
This can be overly restrictive in experimental setups with a high number of parties in close proximity. 
For instance, in a 1$D$ condensed matter system, one could imagine that each site cannot be realistically prohibited from exchanging information carriers with its neighbours during the time window separating the choice of the measurement setting (i.e. the input) and the reading of the measurement outcome (i.e. the output).
As we explain below, even though cluster states are LOSR-$\GMNL$ when no communication is allowed (see Theorem~\ref{theo: 4-cluster}), they can be prepared in only a few rounds of communication between neighbouring sites, and this communication can even be synchronous. 

The above considerations motivate the introduction of the concept of Neighbour-Communication Genuine Multipartite Nonlocality (LONC-$\GMNL$), which is inspired by synchronous distributed computing~\cite{linial1987distributive,linial1992locality,gavoille2009can,le2019quantum}. 
As detailed in Fig.~\ref{fig:CausalExplanatoryModels}c, our LONC model is based on a communication graph along which the parties, represented by nodes, can communicate with their neighbours during $t$ synchronous communication steps (of unbounded bandwidth), and have access to shared randomness. 
Parties can transmit classical, quantum, or, more generally, any \emph{causal resource}: the communication is only restricted by constraints of the no-signalling type~\cite{gisin2020constraints,PhysRevA.104.052207}.
The edges of the graph, which can be oriented or not, indicate the communication flow. 
Importantly, the parties receive their measurement inputs \emph{before} the communication rounds, so inputs can be leaked to neighbours residing up to a distance $t$ away.

Hence, the notion of LONC-$\GMNL$ is always relative to a specific communication graph, based on which the notion of neighbour communication is defined.
This graph should, in practice, be selected according to the experiment topology. 
For simplicity, in the following, we focus on the 1$D$-local scenario, where all parties form a line and where communication flows in one direction only.

\emph{All cluster states are only LONC-$\GMNL_2$; but GHZ states are maximally LONC-$\GMNL$---}
We show in Appendix \ref{app:c-gmnl} that two rounds of one-way quantum communication on a directed path are enough to \emph{prepare} a cluster state $\ket{C_n}$: 
\begin{theorem}\label{clusterline}
All cluster states $\ket{C_n}$ can be prepared in two rounds of one-way synchronous quantum communication.
\end{theorem}
This strengthens a result in~\cite{le2019quantum}, where two-way communication was considered, and it shows that, while $\ket{C_n}$ is hard in the model of Fig.~\ref{fig:CausalExplanatoryModels}b, it is easy in that of Fig.~\ref{fig:CausalExplanatoryModels}c.

Note that Theorem~\ref{clusterline} even allows one to \emph{prepare} $\ket{C_n}$ with \emph{quantum} communication. 
Hence, it only uses a very restricted part of what our LONC-$\GMNL$ model allows for: more generally, super-quantum communication can be used (as long as it does not allow for signalling), and it is sufficient to merely \emph{simulate} the correlations (e.g. by leaking the input).
For instance, it is evident that $\ket{{\rm GHZ}_n}$ cannot be constructed in the same way with less than $n-1$ communication steps (because $\ket{{\rm GHZ}_n}$ is a coherent superposition of perfect correlations between all parties); it is, however, not straightforward to show that $\ket{{\rm GHZ}_n}$ is hard in the LONC-$\GMNL$ model.
In particular, the leaking of inputs up to a distance $t$ offers many ways to simulate some $\ket{{\rm GHZ}_n}$ correlations without necessarily constructing the state.

Despite these difficulties, we conclude this letter by showing that $\ket{{\rm GHZ}_n}$ is maximally LONC-$\GMNL$ for the oriented path graph.

\begin{theorem}\label{GHZline}
$\ket{{\rm GHZ}_n}$ produces correlations that are incompatible with the LONC model along an oriented path with ${t<n-1}$ communication steps.
\end{theorem}
\begin{proof}
    We prove this theorem in Appendix~\ref{sectionGHZonaline}, where we provide a Bell-like inequality that holds for correlations produced with ${t<n-1}$ communication steps on the path ${A^{(1)}\rightarrow \dots\rightarrow A^{(n)}}$, but that is violated by suitable measurements of the $\ket{{\rm GHZ}_n}$ state. 
\end{proof}

\emph{Discussion---}
In this letter, we presented three definitions of genuinely multipartite nonlocality, based on three different frameworks (i.e. LOCC, LOSR, and LONC) and on which causal explanatory model is to be rejected (see Fig.~\ref{fig:CausalExplanatoryModels}).
After discussing the contradictions inherent to the standard Svetlichny LOCC definition (Fig.~\ref{fig:CausalExplanatoryModels}a), we turned to the model of Fig.~\ref{fig:CausalExplanatoryModels}b and showed in a noise-robust way that in this LOSR definition, caterpillar graph states are LOSR-$\GMNL$. 
We then proposed the new LONC model of Fig.~\ref{fig:CausalExplanatoryModels}c, which allows communication between nearest neighbours. We showed that, in the directed path graph, cluster graph states are trivial while the $\ket{{\rm GHZ}_n}$ state is maximally LONC-$\GMNL$.

While we only discussed the theory-agnostic DI setting (which does not limit the underlying causal explanatory model of Fig.~\ref{fig:CausalExplanatoryModels}  to quantum mechanics, and where, for example, nonlocal boxes could be transmitted), it is straightforward to adapt our results and concepts to the DD setting considered in~\cite{Navascu_s_2020, Makuta2023-kj, wang2022quantum}.

Let us conclude with a discussion of experimental benchmarking based on the concept of $\GMNL$, often used to assess the large, nonclassical behavior of experimental systems~\cite{Riedel2010GMEAtomChip,Lucke2014GMEDickeState, Cao2023GMESupercondQubits,Bornet2023GMESpinSqueezing}.
As we have discussed, many different definitions of $\GMNL$ can be considered, leading to radically different conclusions. 
This begs the question, ``\emph{Which definition is the right one?}'' 
We are convinced that no definitive answer can be given \emph{a priori}.
To select the most appropriate framework, one must first carefully reflect on the experimental setup under observation and on which causal explanatory model of the experimental observations should be ruled out.
As we discussed, Svetlichny's LOCC definition, designed to reject the causal explanatory model of Fig.~\ref{fig:CausalExplanatoryModels}a, is often ill-adapted as it fails to reject causal explanations based on bipartite sources.
While an LOSR definition based on the causal explanatory model of Fig.~\ref{fig:CausalExplanatoryModels}b is appropriate under strict conditions of no-signalling, experimentalists working with many subsystems in close proximity might be more interested in refuting models allowing communication between nearest neighbours, as in the LONC causal explanatory model of Fig.~\ref{fig:CausalExplanatoryModels}c.
Some other experimental situations might require the definition of a novel, custom model. 
Hence, before discussing the nature (or depth) of the genuine multipartite entanglement and genuinely multipartite nonlocality of a quantum system produced by some experiment, it is essential to do some rigorous conceptual work to identify a framework and definition that justify, in alignment with the experimental setup, which causal explanatory model is sensible.

\emph{Acknowledgment---}
We thank Otfried Gühne, Simon Perdrix and Rob Thew for helpful discussions. XCR acknowledges funding from the Swiss National Science Foundation. OM acknowledges funding from Competition for Young Scientists 2023. FC acknowledges funding from the Government of Spain (Severo Ochoa CEX2019-000910-S and TRANQI), Fundació Cellex, Fundació Mir-Puig, Generalitat de Catalunya (CERCA program), ERC AdG CERQUTE and Ayuda PRE2022-101448 financiada por MCIN/AEI/ 10.13039/501100011033 y por el FSE+. MOR acknowledges funding by INRIA Action Exploratoire project
DEPARTURE and from ANR JCJC LINKS (ANR-23-CE47-0003). OM and RA acknowledge support by the (Polish) National Science Center through the SONATA BIS Grant No. 2019/34/E/ST2/00369.

\bibliography{refs}

\newpage
\onecolumngrid
\appendix
\section{Introduction to graph states}\label{app:graph_states}

\begin{figure}[b]
    {\centering
\resizebox{2\textwidth/3}{!}{\definecolor{myred}{RGB}{204,51,17}
\definecolor{myblue}{RGB}{0,119,187}
\definecolor{mygrey}{RGB}{135,135,135}
\definecolor{myteal}{RGB}{80,175,148}
\definecolor{myorange}{RGB}{238,119,51}
\definecolor{mymagenta}{RGB}{238,51,110}

\tikzset{player/.style={circle,draw=black,inner sep=1, minimum size=15pt,fill=white,thick
 }}

\tikzset{ligne/.style={very thick
}}

\centering
\begin{tikzpicture}[scale=1]
\centering

\usetikzlibrary{decorations.pathreplacing}
\tikzset{ligne/.style={very thick,myred
}}
\tikzset{ligne2/.style={very thick,black
}}

\node[player] (A) at (0,0) {$A$};

\node[player] (B) at (-4/3,5/3) {$B$};
\node[player] (C) at (0,2) {$C$};
\node[player] (D) at (1,4/3) {$D$};

\draw[ligne2] (B) -- (A) -- (C);
\draw[ligne2] (A) -- (D);
\draw[ligne] (B) -- (C);

\begin{scope}[shift={(8,0)}]
\node[player] (A) at (0,0) {$A$};
\node[player] (B) at (-4/3,5/3) {$B$};
\node[player] (C) at (0,2) {$C$};
\node[player] (D) at (1,4/3) {$D$};

\draw[ligne2] (B) -- (A) -- (C);
\draw[ligne2] (A) -- (D);
\draw[ligne] (B) -- (D)--(C);
\end{scope}

\draw[ligne2,very thick,double,{<[length=2mm, open]}-{>[length=2mm, open]}] (2,0.8) -- (6,0.8);
\node[align=left,anchor=south] at (4,1) {\large local complementation};

\begin{scope}[shift={(0,-3.2)}]
\node[player] (A) at (0,0) {$A$};
\node[player] (B) at (-4/3,5/3) {$B$};
\node[player] (C) at (0,2) {$C$};
\node[player] (D) at (1,4/3) {$D$};

\draw[ligne2] (B) -- (A);
\draw[ligne2] (A) -- (D);
\draw[ligne2] (A) to[out = 20, in = -90] (D);
\draw[ligne2] (A) to[out = 70, in = 210] (D);

\draw[ligne2] (C) to[out = -120, in = 120] (A);
\draw[ligne2] (C) to[out = -100, in = 100] (A);

\begin{scope}[shift={(8,0)}]
\node[player] (A) at (0,0) {$A$};
\node[player] (B) at (-4/3,5/3) {$B$};
\node[player] (C) at (0,2) {$C$};
\node[player] (D) at (1,4/3) {$D$};

\draw[ligne2] (B) to[out = -80, in = 160] (A);
\draw[ligne2] (B) to[out = -50, in = 140] (A);

\draw[ligne2] (C) to[out = -60, in = 60] (A);
\draw[ligne2] (C) to[out = -80, in = 80] (A);
\draw[ligne2] (C) to[out = -100, in = 100] (A);
\draw[ligne2] (C) to[out = -120, in = 120] (A);

\draw[ligne2] (A) -- (D);
\end{scope}

\draw[ligne2,very thick,double,{[length=2mm, open]}-{>[length=2mm, open]}] (2,0.8) -- (6,0.8);
\node[align=left,anchor=south] at (4,1) {\large vertex multiplication by $2$};
\node[align=left,anchor=west] at (-2,2.5) {\large b.};
\end{scope}

\node[align=left,anchor=west] at (-2,2.5) {\large a.};

\end{tikzpicture}}
   } \caption{
    \begin{minipage}[t]{0.6\linewidth} a. Local complementation acting on $A$, for dimension $d=2$.\\ b. Vertex multiplication by $2$ acting on $A$, for dimension $d=5$.\end{minipage}
   }\label{fig:graph_local_operations}
\end{figure}
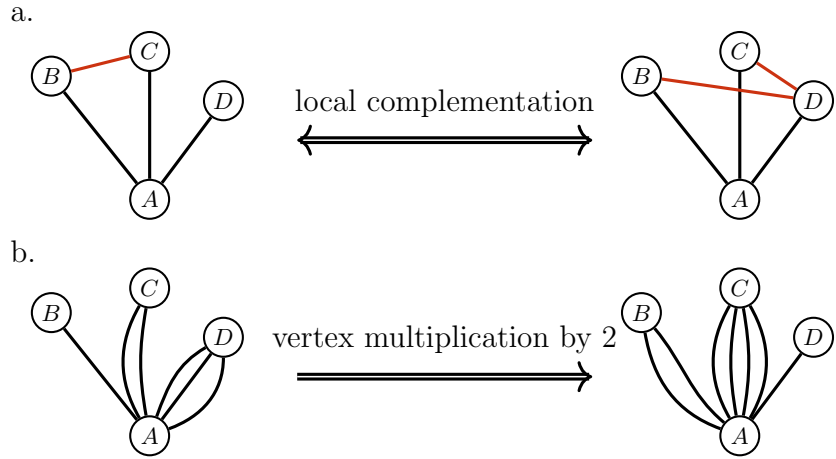
In this section, we give a comprehensive introduction to the graph states relevant from the perspective of our work. However, to properly define graph states, we first have to introduce the basic notions of graph theory. A graph $G$ is defined as a pair of sets $G=(V,E)$ where $V$ is a set of elements called vertices, while $E$ is the set consisting of edges connecting elements from $V$. In this work, we make the most commonly used assumption that each individual edge connects exactly two, different vertices. However, in contrast to what is usually assumed, we consider graphs in which two vertices can be connected by more than one edge. Such graphs are commonly referred to as multigraphs, but here we do not make that distinction and refer to them simply as graphs. 

Among the many properties of graphs that are commonly studied, three are of particular interest in the context of our work. The first one is connectivity: we say that a graph is connected if there exists a path from any vertex to any other vertex in the graph, i.e. we can get from any vertex to any other vertex by traveling along the edges of the graph.

The second definition of interest to us is that of an induced subgraph $G_{S}=(V_{S},E_{S})$ of a graph $G=(V,E)$, which is defined by the following relations:
\begin{equation}
V_{S}\subseteq V\,, \qquad E_{S}=\{ \{i,j\}\in E | i,j \in V_{S} \}\,.
\end{equation}
In descriptive terms, an induced subgraph $G_{S}$ is constructed from the graph $G$ by removing some vertices of $G$ and all edges related to the removed vertices. The third important graph property is the neighbourhood $\mathcal{N}_{i}$ of a vertex $i$. It is defined as a set of all vertices in $V$ connected to $i$ by at least one edge
\begin{equation}\label{eq:neighbourhood}
\mathcal{N}_{i}=\{j\in V | \{i,j\}\in E\}\,.
\end{equation}

Using these three properties, we can define the classes of graphs on which we focus in our work. The first class, called linear graphs (or path graphs), is pretty self-explanatory: it is a graph that forms a line. In more technical terms, a graph $G=(V,E)$ is a linear graph if it is connected and if, for exactly two vertices $i,j\in V$, we have $|\mathcal{N}_{i}|=|\mathcal{N}_{j}|=1$, while $|\mathcal{N}_{k}|=2$ for all $k\in V\setminus \{i,j\}$. In order to define the second class, we first need to consider the induced subgraph $G_{S}=(V_{S},E_{S})$ of $G=(V,S)$ defined by 
\begin{equation}\label{eq:spine}
V_{S}=\{i\in V| |\mathcal{N}_{i}|\geqslant 2\}\,.
\end{equation}
If $G$ is connected and the induced subgraph $G_{S}$ is a linear graph, then we call $G$ a caterpillar graph. In other words, a caterpillar graph is a path graph with additional vertices attached to it, such that the only vertex in their neighbourhood is one of the vertices in the path graph.  See Fig.~\ref{fig:graphstates} for an example of these two classes.

While the formal definition of a graph and its graphical description both prove useful for some applications, they are unsuitable for an efficient mathematical analysis of a graph structure. To that end, one considers an adjacency matrix $\Gamma$ of a graph $G$. The matrix $\Gamma$ is defined element-wise: a matrix element $\Gamma_{i,j}$ corresponds to the number of edges $\{i,j\}\in E$, connecting vertices $i,j\in V$, where we label vertices as $V=\{1,2,\dots,N \}$. 

There are three properties of $\Gamma$ worth highlighting. First, each element $\Gamma_{i,j}$ has to be a natural number (including zero), as vertices can be connected only by a discrete number of edges, but this number can be larger than one, as we consider multigraphs. Second, $\Gamma$ has to be symmetric, as the number of edges connecting $i$ and $j$ is the same as the number of edges connecting $j$ and $i$. Lastly, since we exclude edges that connect a vertex to itself, we have $\Gamma_{i,i}=0$ for all $i\in V$.

With that, we are ready to define graph states. Let us consider a Hilbert space $\mathcal{H}=\mathbb{C}_{d}^{\otimes N}$ and generalised Pauli matrices
\begin{equation}\label{eq:def_xz}
X= \sum_{j=0}^{d-1} \ket{j+1}\! \bra{j}\,, \qquad Z=\sum_{j=0}^{d-1}\omega^{j} \ket{j}\! \bra{j}\,,
\end{equation}
where $\omega=\exp(2 \pi \mathbb{i}/d)$ and $\ket{d} \equiv\ket{0}$. Given $\mathcal{H}$ and a graph $G$, a graph state $\ket{G}\in \mathcal{H}$ is defined as the unique state fulfilling
\begin{equation}
g_{i} \ket{G} = \ket{G}
\end{equation}
for all $i$, where
\begin{equation}\label{eqn: gi_def}
g_{i} = X_{i} \prod_{j=1}^{N} Z_{j}^{\Gamma_{i,j}}
\end{equation}
and $X_{i}$ is an operator acting with a generalised Pauli matrix $X$ on the $i$'th subsystem (qudit) and with an identity operator $\mathbb{1}$ on every other qudit (similarly for $Z_{j}$). Note that only the number of edges $\Gamma_{i,j} \mod d$ matters in \eqref{eqn: gi_def}, so we are free to restrict our consideration to graphs for which $\Gamma_{i,j} \leqslant d-1$ holds for all $i, j$.

One of the most important properties of graph states is that some of them are equivalent to each other via local unitary operations. 
For example, we point out that the state $\ket{C_4}$ described in the main text (and repeated in Eq.~\eqref{eqn: 4cluster_state} in Appendix~\ref{app: 4-cluster}) is described by a set of stabilizers equivalent to $\mathbb{S}=\{ ZZII, XXZI, IZXX, IIZZ\}$, whereas applying Eq.~\eqref{eqn: gi_def} would define the four-party linear graph state as stabilized by $\mathbb{S}'=\{ XZII, ZXZI, IZXZ, IIZX\}$. However, $\mathbb{S}$ and $\mathbb{S}'$, and therefore their associated states, are equivalent up to the local unitary $U=HIIH$, where $H$ is the Hadamard matrix. (Note that in the above, we have omitted the tensor product symbol $\otimes$ to lighten the notation.)

Interestingly, the equivalence under local unitaries can also be studied on the graphs themselves, i.e. for some graph transformations $G\rightarrow G'$ we know that there exists a local unitary operation $LU: \ket{G}\rightarrow \ket{G'}$ (see \cite{bahramgiri2007graph} for more details on this topic).  

The first such operation is called \emph{local complementation} (see Fig.~\ref{fig:graph_local_operations}a). This operation acts on a chosen vertex $n\in V$ such that the edges in the graph are modified according to the following formula 
\begin{equation}
\Gamma_{i,j}'=\Gamma_{i,j} + \Gamma_{i,n}\Gamma_{j,n} \mod d\,,
\end{equation}
where $\Gamma'$ is the adjacency matrix of the graph after the transformation. Notice, that this changes only the neighbourhood of $n$ (or more formally, the induced subgraph $G_{S}=(V_{S}=\mathcal{N}_{n},E_{S})$) while keeping the rest of the edges in the graph unchanged.

The second transformation of this type is \emph{vertex multiplication} (see Fig.~\ref{fig:graph_local_operations}b), whose action on a vertex $n$ is described by
\begin{equation}
\Gamma_{i,n}'= b \Gamma_{i,n} \mod d \quad \textrm{for all }i\in V\,, \qquad \Gamma_{i,j}'=\Gamma_{i,j} \quad \textrm{for all }i,j\neq n\,,
\end{equation}
where $b\in\{1,\dots,d-1\}$ can be chosen freely.

Using both of these transformations, one can find graph states that are related to each other via local unitaries. While this is a very difficult task in general, for some examples it is fairly straightforward. Let us consider one of the simplest examples relevant to our results: let us show that, given that $d$ is prime, every graph state corresponding to a linear graph, where the number of vertices $N$ can vary, is equivalent up to local unitaries to a graph state corresponding to a linear graph in which two consecutive vertices are connected by a single edge (called cluster state).

Without loss of generality, we can take that
\begin{equation}
\Gamma_{i,i+1} = \Gamma_{i+1,i} \neq 0 \mod d \quad \textrm{for all  }i\in\{1,\dots,N-1\}\,,
\end{equation}
and $\Gamma_{i,j}=0$ for all other $i$ and $j$. Let us start by acting on vertex $2$ with vertex multiplication such that
\begin{equation}
\Gamma_{1,2}' = b\Gamma_{1,2} = 1 \mod d.
\end{equation}
Notice that we can always find such a $b$ as long as $d$ is prime. Next, we act on vertex $3$ with vertex multiplication such that
\begin{equation}
\Gamma_{2,3}' = b\Gamma_{2,3} = 1 \mod d\,.
\end{equation}
We can repeat this procedure until we get
\begin{equation}
\Gamma_{i,i+1} = 1 \quad \textrm{for all } i\in \{1,\dots, N-2\}\,.
\end{equation}
Lastly, to make $\Gamma_{N-1,N}' =1$, we act with an appropriate vertex multiplication on the vertex $N$.

\section{LOSR-GMNL detection for a 4-qubit cluster state}\label{app: 4-cluster}
\begin{figure}[h]
\centering
\resizebox{\columnwidth/2}{!}{%
\begin{tikzpicture}[scale=0.18]
\centering

\node[player] (A0) at (0,10) {$A$};
\node[player] (B0) at (10,10) {$B$};
\node[player] (C0) at (0,0) {$C$};
\node[player] (D0) at (10,0) {$D$};

\node[sourcea] (SA0) at (10/3,20/3) { };
\node[sourceb] (SB0) at (20/3,20/3) { };
\node[sourcec] (SC0) at (10/3,10/3) { };
\node[sourced] (SD0) at (20/3,10/3) { };
\begin{scope}[on background layer]
\draw[ligne] (A0) -- (SA0) (A0) -- (SB0) (A0) -- (SC0);
\draw[ligne] (B0) -- (SA0) (B0) -- (SB0) (B0) -- (SD0);
\draw[ligne] (C0) -- (SA0) (C0) -- (SD0) (C0) -- (SC0);
\draw[ligne] (D0) -- (SD0) (D0) -- (SB0) (D0) -- (SC0);
\end{scope}
\node at (5,-4) {\Large $\mathcal{I}_0$};

\draw[fleche] (12,5) -- (16,5) ;
\draw[fleche] (32,-15) -- (36,-15) ;

\draw[fleche] (32,5) -- (36,5) ;
\draw[fleche] (12,-15) -- (16,-15) ;
\draw[fleche] (53,5) -- ({53+4*cos(35)},{5-4*sin(35}) ;
\draw[fleche] (53,-20+5) -- ({53+4*cos(35)},{-20+5+4*sin(35}) ;

\begin{scope}[shift={(20,0)}]
\node[player] (A0) at (0,10) {$A$};
\node[player] (B0) at (10,10) {$B$};
\node[player] (C0) at (0,0) {$C$};
\node[player] (D0) at (10,0) {$D$};

\node[sourcea] (SA0) at (10/3,20/3) { };
\node[sourceb] (SB0) at (20/3,20/3) { };
\node[sourcec] (SC0) at (10/3,10/3) { };
\node[sourcec] (SC1) at (0,5) { };
\node[sourced] (SD0) at (20/3,10/3) { };
\begin{scope}[on background layer]
\draw[ligne] (A0) -- (SA0) (A0) -- (SB0) (A0) -- (SC1);
\draw[ligne] (B0) -- (SA0) (B0) -- (SB0) (B0) -- (SD0);
\draw[ligne] (C0) -- (SA0) (C0) -- (SD0) (C0) -- (SC0);
\draw[ligne] (D0) -- (SD0) (D0) -- (SB0) (D0) -- (SC0);
\end{scope}
\node at (5,-4) {\Large $\mathcal{I}_1$};
\end{scope}

\begin{scope}[shift={(40,0)}]
\node[player] (A0) at (0,10) {$A$};
\node[player] (B0) at (10,10) {$B$};
\node[player] (C0) at (0,0) {$C$};
\node[player] (D0) at (10,0) {$D$};

\node[sourcea] (SA0) at (10/3,20/3) { };
\node[sourceb] (SB0) at (20/3,20/3) { };
\node[sourcec] (SC0) at (10/3,10/3) { };
\node[sourcec] (SC1) at (0,5) { };
\node[sourced] (SD0) at (20/3,10/3) { };
\node[sourced] (SD1) at (10,5) { };
\begin{scope}[on background layer]
\draw[ligne] (A0) -- (SA0) (A0) -- (SB0) (A0) -- (SC1);
\draw[ligne] (B0) -- (SA0) (B0) -- (SB0) (B0) -- (SD1);
\draw[ligne] (C0) -- (SA0) (C0) -- (SD0) (C0) -- (SC0);
\draw[ligne] (D0) -- (SD0) (D0) -- (SB0) (D0) -- (SC0);
\end{scope}
\node at (5,-4) {\Large $\mathcal{I}_2$};
\end{scope}

\begin{scope}[shift={(0,-10)},yscale=-1]
\node[player] (A0) at (0,10) {$A'$};
\node[player] (B0) at (10,10) {$B'$};
\node[player] (C0) at (0,0) {$C$};
\node[player] (D0) at (10,0) {$D$};

\node[sourcea] (SA0) at (10/3,20/3) { };
\node[sourceb] (SB0) at (20/3,20/3) { };
\node[sourcec] (SC0) at (10/3,10/3) { };
\node[sourced] (SD0) at (20/3,10/3) { };
\begin{scope}[on background layer]
\draw[ligne] (A0) -- (SA0) (A0) -- (SB0) (A0) -- (SC0);
\draw[ligne] (B0) -- (SA0) (B0) -- (SB0) (B0) -- (SD0);
\draw[ligne] (C0) -- (SA0) (C0) -- (SD0) (C0) -- (SC0);
\draw[ligne] (D0) -- (SD0) (D0) -- (SB0) (D0) -- (SC0);
\end{scope}
\node at (5,14) {\Large $\mathcal{I}_3$};
\end{scope}

\begin{scope}[shift={(20,-10)},yscale=-1]
\node[player] (A0) at (0,10) {$A'$};
\node[player] (B0) at (10,10) {$B'$};
\node[player] (C0) at (0,0) {$C$};
\node[player] (D0) at (10,0) {$D$};

\node[sourcea] (SA0) at (10/3,20/3) { };
\node[sourcea] (SA1) at (0,5) { };
\node[sourceb] (SB0) at (20/3,20/3) { };
\node[sourcec] (SC0) at (10/3,10/3) { };
\node[sourced] (SD0) at (20/3,10/3) { };
\begin{scope}[on background layer]
\draw[ligne] (A0) -- (SA0) (A0) -- (SB0) (A0) -- (SC0);
\draw[ligne] (B0) -- (SA0) (B0) -- (SB0) (B0) -- (SD0);
\draw[ligne] (C0) -- (SA1) (C0) -- (SD0) (C0) -- (SC0);
\draw[ligne] (D0) -- (SD0) (D0) -- (SB0) (D0) -- (SC0);
\end{scope}
\node at (5,14) {\Large $\mathcal{I}_4$};
\end{scope}

\begin{scope}[shift={(40,-10)},yscale=-1]
\node[player] (A0) at (0,10) {$A'$};
\node[player] (B0) at (10,10) {$B'$};
\node[player] (C0) at (0,0) {$C$};
\node[player] (D0) at (10,0) {$D$};

\node[sourcea] (SA0) at (10/3,20/3) { };
\node[sourcea] (SA1) at (0,5) { };
\node[sourceb] (SB0) at (20/3,20/3) { };
\node[sourceb] (SB1) at (10,5) { };
\node[sourcec] (SC0) at (10/3,10/3) { };
\node[sourced] (SD0) at (20/3,10/3) { };
\begin{scope}[on background layer]
\draw[ligne] (A0) -- (SA0) (A0) -- (SB0) (A0) -- (SC0);
\draw[ligne] (B0) -- (SA0) (B0) -- (SB0) (B0) -- (SD0);
\draw[ligne] (C0) -- (SA1) (C0) -- (SD0) (C0) -- (SC0);
\draw[ligne] (D0) -- (SD0) (D0) -- (SB1) (D0) -- (SC0);
\end{scope}
\node at (5,14) {\Large $\mathcal{I}_5$};
\end{scope}

\begin{scope}[shift={(180/3,-5)}]
\node[player] (A0) at (0,10) {$A$};
\node[player] (B0) at (10,10) {$B$};
\node[player] (C0) at (0,0) {$C$};
\node[player] (D0) at (10,0) {$D$};

\node[sourcea] (SA0) at (10/3,20/3) { };
\node[sourceb] (SB0) at (20/3,20/3) { };
\node[sourcec] (SC1) at (0,5) { };
\node[sourced] (SD1) at (10,5) { };
\begin{scope}[on background layer]
\draw[ligne] (A0) -- (SA0) (A0) -- (SB0) (A0) -- (SC1);
\draw[ligne] (B0) -- (SA0) (B0) -- (SB0) (B0) -- (SD1);
\draw[ligne] (C0) -- (SA0) ;
\draw[ligne] (D0) -- (SB0) ;
\end{scope}
\begin{scope}[yscale=-1]
\node[player] (A0) at (0,10) {$A'$};
\node[player] (B0) at (10,10) {$B'$};

\node[sourcea] (SA0) at (10/3,20/3) { };
\node[sourceb] (SB0) at (20/3,20/3) { };
\node[sourcec] (SC0) at (10/3,10/3) { };
\node[sourced] (SD0) at (20/3,10/3) { };
\begin{scope}[on background layer]
\draw[ligne] (A0) -- (SA0) (A0) -- (SB0) (A0) -- (SC0);
\draw[ligne] (B0) -- (SA0) (B0) -- (SB0) (B0) -- (SD0);
\draw[ligne] (C0) -- (SD0) (C0) -- (SC0);
\draw[ligne] (D0) -- (SD0) (D0) -- (SC0);
\end{scope}
\end{scope}

\node at (5,-14.5) {\Large $\mathcal{J}$};

\end{scope}

\end{tikzpicture}}

FIG.~\ref{fig:intuitionc4} from the main text. It represents the inflation method used here in Appendix~\ref{app: 4-cluster}.
\end{figure}

In this section, we derive a proof of LOSR-$\GMNL$ for an informative example of the $4$-qubit cluster state, in order to show readers unfamiliar with the inflation technique how it works. Let us begin by recalling the monogamy relation from \cite[Theorem 1]{Augusiak_2014},
\begin{align}
I^{\tilde{C}D}+2\langle \tilde{C}_0 E_{0}\rangle \leqslant 4\,,
\end{align}
where $I^{\tilde{C}D}$ is the CHSH inequality \cite{PhysRevLett.23.880}
\begin{equation}
I^{\tilde{C}D} = \langle \tilde{C}_0 D_0 \rangle + \langle \tilde{C}_0 D_1 \rangle + \langle \tilde{C}_1 D_0\rangle -\langle \tilde{C}_1 D_1 \rangle\
\end{equation}
and $\tilde{C}_{0}, \tilde{C}_{1}$ are dichotomic observables with eigenvalues $\pm 1$ representing measurements of the party $\tilde{C}$  (likewise for $D_{0},D_{1},E_{0}$). By substituting $\tilde{C}_{0}=C_{0}$, $\tilde{C}_{1}=A_{0}C_{1}$, $E_{0}=A'_{1}B'_{1}$, where $A_{i}, A'_{i}$ are observables belonging to parties $A$ and $A'$ respectively, we get the following inequality
\begin{equation}\label{eq:monogamy_example}
I^{ACD} + 2 \langle A'_{1} B'_{1} C_{0}\rangle \leqslant 4\,,
\end{equation}
where
\begin{align}
I^{\tilde{C}D}=I^{ACD}= \langle C_0 D_0 \rangle + \langle C_0 D_1 \rangle + \langle A_0 C_1 D_0\rangle -\langle A_0 C_1 D_1 \rangle\,.
\end{align}

Let us now consider the inflations presented in Fig.~\ref{fig:intuitionc4}. From the assumption that parties in a network cannot communicate with each other classically, we can conclude that the monogamy relation \eqref{eq:monogamy_example} has to hold for any correlation originating from the inflation $\mathcal{J}$. The task now is to express the above inequality in terms of correlations originating from the original network, represented in Fig.~\ref{fig:intuitionc4} by $\mathcal{I}_{0}$.

Let us begin by considering the CHSH-like term $I^{ACD}$ calculated for the correlations originating from the inflation $\mathcal{I}_{1}$. Since we assume that each copy of a party performs the same measurements and each copy of a source generates the same correlations, if the value of $I^{ACD}$ is different for two inflations, it has to be due to the different network structures between two inflations. Moreover, notice that if the value $I^{ACD}$ depended on which sources $B$ is connected to, $B$ could manipulate this value by choosing to perform measurements on some local correlations rather than correlations sent by the sources in the network. This way $B$ could send information to the rest of the parties, and since we assume that parties cannot communicate with each other, we can conclude that the value of $I^{ACD}$ can only depend on the structure of the subnetwork of parties $A, C$ and $D$. 

Consequently, since the subnetwork of parties $A, C$ and $D$ in $\mathcal{J}$ is equivalent to the same subnetwork in $\mathcal{I}_{1}$, it follows that the value of $I^{ACD}$ is the same over correlations originating from $\mathcal{J}$ and $\mathcal{I}_{1}$:
\begin{equation}
I^{ACD}_{\mathcal{J}} = I^{ACD}_{\mathcal{I}_{1}}\,.
\end{equation}

Furthermore, looking at each expected value in $I^{ACD}$ individually, we can see that for $\langle C_{0}D_{0}\rangle$ and $\langle C_{0}D_{1}\rangle$ the only relevant subnetwork is actually the subnetwork of parties $C$ and $D$, and since this subnetwork is equivalent in $\mathcal{J}$, $\mathcal{I}_{1}$ and $\mathcal{I}_{0}$ we have
\begin{equation}
\langle C_{0}D_{j}\rangle_{\mathcal{J}}=\langle C_{0}D_{j}\rangle_{\mathcal{I}_{1}}=\langle C_{0}D_{j}\rangle_{\mathcal{I}_{0}}
\end{equation}
for $j\in\{0,1\}$. However, the same cannot be said for the expected values $\langle A_{0}C_{1}D_{0}\rangle$ and $\langle A_{0}C_{1}D_{1}\rangle$, since the subnetwork of parties $A$, $C$ and $D$ in $I_{0}$ is inequivalent to the same subnetwork in $\mathcal{I}_{1}$.  Luckily, we can overcome this problem using an inequality from \cite{Navascu_s_2020}
\begin{align}
\langle A_0 C_1 D_0 \rangle_{\mathcal{I}_{1}} \geqslant \langle A_0 B_0 \rangle_{\mathcal{I}_{1}} + \langle B_0 C_1 D_0 \rangle_{\mathcal{I}_{1}} -1\,,
\end{align}
(see Lemma \ref{lem:ineq} in Appendix \ref{app:C_N} for a general formulation of this inequality).

Now, we can again argue that the subnetwork of parties $A$ and $B$ is equivalent in both $\mathcal{I}_{1}$ and $\mathcal{I}_{0}$, and that the subnetwork of parties $B$, $C$ and $D$ is also equivalent in these two inflations, hence
\begin{equation}
\langle A_0 B_0 \rangle_{\mathcal{I}_{1}} = \langle A_0 B_0 \rangle_{\mathcal{I}_{0}}\,,\qquad \langle B_0 C_1 D_0 \rangle_{\mathcal{I}_{1}} =  \langle B_0 C_1 D_0 \rangle_{\mathcal{I}_{0}}\,.
\end{equation}
We can follow the same argument to get 
\begin{align}
-\langle A_0 C_1 D_1 \rangle_{\mathcal{I}_{1}} \geqslant \langle A_0 B_0 \rangle_{\mathcal{I}_{0}} - \langle B_0 C_1 D_1 \rangle_{\mathcal{I}_{1}} -1\,.
\end{align}

Putting it all together, we conclude that the value of $I^{ACD}_{\mathcal{J}}$ is bounded from below by
\begin{align}
I^{ACD}_{\mathcal{J}} \geqslant I^{BCD}_{\mathcal{I}_0}+2 \langle A_0 B_0 \rangle_{\mathcal{I}_0}-2\,,
\end{align}
where
\begin{equation}\label{ineq:example_chsh_bound}
I^{BCD} = \langle C_{0}D_{0}\rangle + \langle C_{0}D_{1}\rangle + \langle B_0 C_1 D_0 \rangle- \langle B_0 C_1 D_1 \rangle\,.
\end{equation}

Let us now look at the term $\langle C_{0} A'_1 B'_1\rangle$ in \eqref{eq:monogamy_example}. Since the subnetwork of parties $A'$, $B'$ and $C$ is the same in $\mathcal{J}$ as in $\mathcal{I}_{4}$, we have
\begin{equation}
\langle C_0 A'_1 B'_1 \rangle_{\mathcal{J}}=\langle C_0 A'_1 B'_1 \rangle_{\mathcal{I}_{4}}\,.
\end{equation}
Next, we can use the inequality from \cite{Navascu_s_2020} to get the following:
\begin{align}
  \langle C_0 A'_1 B'_1 \rangle_{\mathcal{I}_{4}} \geqslant \langle  C_0 D_2 \rangle_{\mathcal{I}_{4}} + \langle A'_1 B'_1 D_2 \rangle_{\mathcal{I}_{4}} -1\,. 
\end{align}
Since the subnetworks of $C$ and $D$ and of $A'$, $B'$ and $D$ are the same in $\mathcal{I}_{3}$ as in $\mathcal{I}_4$, we have
\begin{equation}
    \langle  C_0 D_2 \rangle_{\mathcal{I}_{4}} = \langle  C_0 D_2 \rangle_{\mathcal{I}_{3}}\,,\qquad \langle A'_1 B'_1 D_2 \rangle_{\mathcal{I}_{4}} = \langle A'_1 B'_1 D_2 \rangle_{\mathcal{I}_{3}}\,.
\end{equation}

Moreover, since $\mathcal{I}_{3}$ can be transformed into $\mathcal{I}_{0}$ by relabeling parties $A'$ and $B'$ to $A$ and $B$, we can conclude that
\begin{equation}
\langle  C_0 D_2 \rangle_{\mathcal{I}_{3}} = \langle  C_0 D_2 \rangle_{\mathcal{I}_{0}}\,,\qquad \langle A'_1 B'_1 D_2 \rangle_{\mathcal{I}_{3}} = \langle A_1 B_1 D_2 \rangle_{\mathcal{I}_{0}}\,,
\end{equation}
therefore
\begin{equation}
    \langle C_0 A'_1 B'_1 \rangle_{\mathcal{J}} \geqslant \langle  C_0 D_2 \rangle_{\mathcal{I}_{0}} + \langle A_1 B_1 D_2 \rangle_{\mathcal{I}_{0}} -1\,. 
\end{equation}
Combining this with Ineq. \eqref{ineq:example_chsh_bound} and \eqref{eq:monogamy_example}, we get the final inequality, that has to be fulfilled for any correlations originating from $\mathcal{I}_{0}$,
\begin{align}\label{ineq:example_final}
    I^{BCD}_{\mathcal{I}_{0}}+ 2\langle A_0 B_0 \rangle_{\mathcal{I}_{0}} + 2 \langle C_0 D_2 \rangle_{\mathcal{I}_{0}} + 2\langle A_1 B_1 D_2 \rangle_{\mathcal{I}_{0}} \leqslant 8\,.
\end{align}

We can now use this inequality to show that the cluster state
\begin{equation}\label{eqn: 4cluster_state}
\ket{C_{4}} = \frac{1}{2}\left(\ket{0000} + \ket{0011}+\ket{1100} - \ket{1111}\right)
\end{equation}
is LOSR-$\GMNL$. To this end, we take the  observables 
\begin{equation}
A_{0}=B_{0}=C_{0}=D_{2}=Z\,,\quad A_{1}=B_{1}=C_{1}=X\,,\quad D_{0}=\frac{1}{\sqrt{2}}(Z+X)\,,\quad D_{1}=\frac{1}{\sqrt{2}}(Z-X)\,.
\end{equation}
Computing the left-hand side of Ineq. \eqref{ineq:example_final} for those observables and a state $\ket{C_{4}}$ gives us the value $2\sqrt{2}+6$ which violates the inequality. This implies that measurement statistics of $\ket{C}_{4}$ cannot be reproduced by the network $\mathcal{I}_{0}$, or, in other words, that $\ket{C_{4}}$ is LOSR-$\GMNL$.

\section{All caterpillar states are LOSR-GMNL}\label{app:C_N}

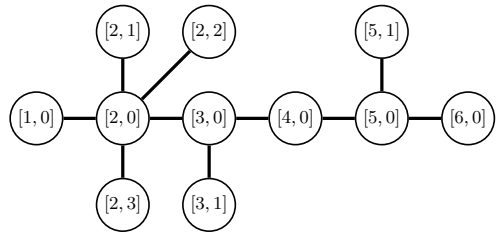
\begin{figure}[b]
    \centering
   \resizebox{0.8\textwidth/2}{!}{\definecolor{myred}{RGB}{204,51,17}
\definecolor{myblue}{RGB}{0,119,187}
\definecolor{mygrey}{RGB}{135,135,135}
\definecolor{myteal}{RGB}{80,175,148}
\definecolor{myorange}{RGB}{238,119,51}
\definecolor{mymagenta}{RGB}{238,51,110}

\tikzset{player/.style={circle,draw=black,inner sep=2, minimum size=15pt,fill=white,thick
 }}

\tikzset{ligne/.style={very thick
}}

\centering
\begin{tikzpicture}[scale=1.5]
\centering

\usetikzlibrary{decorations.pathreplacing}
\tikzset{ligne/.style={ultra thick,myred
}}
\tikzset{ligne2/.style={ultra thick,black
}}

\node[player] (10) at (0,0) {$[1,0]$};
\node[player] (20) at (1,0) {$[2,0]$}; 
\node[player] (30) at (2,0) {$[3,0]$};
\node[player] (40) at (3,0) {$[4,0]$};
\node[player] (50) at (4,0) {$[5,0]$};

\node[player] (21) at (1,1) {$[2,1]$};
\node[player] (21b) at (2,1) {$[2,2]$};
\node[player] (21c) at (1,-1) {$[2,3]$};
\node[player] (31) at (2,-1) {$[3,1]$};
\node[player] (51) at (4,1) {$[5,1]$};
\node[player] (52) at (5,0) {$[6,0]$};

\draw[ligne2] (10) -- (20) -- (30) -- (40) -- (50);
\draw[ligne2] (21b)--(20) -- (21);
\draw[ligne2] (20)--(21c);
\draw[ligne2] (30)--(31);
\draw[ligne2] (52)--(50)--(51);

\end{tikzpicture}}
    \caption{Notation used for the vertices of a caterpillar graph. Note how the extremal vertices $[1,0]$ and $[6,0]$ have degree $1$ by definition.
   }\label{fig:notation_caterpillars}
\end{figure}

In this section, we show that all graph states corresponding to caterpillar graphs (including, as a special case, all cluster states) 
are LOSR-$\GMNL$. 
To this end, we first construct an inequality based on the monogamy relation from \cite[Theorem 1]{Augusiak_2014} that has to be satisfied by any behavior originating from an $N$-partite network with $(N-1)$-partite sources of non-signalling correlations and an $N$-partite source of classical correlations. We then show that there exists a choice of measurements under which caterpillar states violate that inequality.

Let us begin by establishing the notation. An $N$-partite caterpillar state $\ket{\ddagger_{N}}$ is a graph state corresponding to a caterpillar graph with $N$ vertices (see Appendix \ref{app:graph_states} for more details on graph states, including the definition of caterpillar graphs). 

It will be convenient for the purposes of the proof to denote those vertices/parties not with a single number, but with a pair of numbers. To this end, we first need to find the longest induced line subgraph of the caterpillar graph. We denote the vertices belonging to this subgraph by $[i,0]$ for $i\in \{1,\dots, L\}$, where $L$ is the order (the number of vertices) of the line subgraph. We denote vertices that do not belong to said linear graph by $[i,j]$ for $j\geqslant 1$, where $i$ is chosen such that $[i,j]$ is connected by an edge only to the vertex $[i,0]$. Notice that since vertices $\{[1,0],[2,0],\dots,[L,0]\}$ form the longest induced line subgraph, there are no vertices $[i,j]$ such that $i=1,L$ and $j\geqslant 1$. See Fig.~\ref{fig:notation_caterpillars} for a visual example.

Under this notation, the state $\ket{\ddagger_{N}}$ is uniquely defined by the relations $g_{[i,j]}\ket{\ddagger_{N}}=\ket{\ddagger_{N}}$ for
\begin{align}\label{eqn: caterpillar_stabs}
\begin{split}
g_{[1,0]}&=X_{[1,0]}Z_{[2,0]}\,,\\
g_{[i,0]}&=Z_{[i-1,0]}X_{[i,0]}Z_{[i+1,0]}\prod_{j= 1}^{n_{i}}Z_{[i,j]}\quad \textrm{for } i\in\{2,\dots,L-1\}\,,\\
g_{[i,j]}&=Z_{[i,0]}X_{[i,j]}\qquad \textrm{for } j\in \{1,\dots,n_{i}\}\,,\\
g_{[L,0]}&=Z_{[L-1,0]}X_{[L,0]}\,,
\end{split}
\end{align}
where $X_{[i,j]}, Z_{[i,j]}$ are Pauli operators acting on the $[i,j]$'th qubit (see Eq. \eqref{eq:def_xz} for $d=2$), $n_{i}=|\mathcal{N}_{[i,0]}|-2$ and $\mathcal{N}_{[i,0]}$ is the neighbourhood of $[i,0]$ (see Appendix \ref{app:graph_states} for the definition). 

The scenario considered in the proof is the following: an $N$-partite network $\mathcal{O}$ with $N$, $(N-1)$-partite sources of non-signalling correlations and an $N$-partite source of classical correlations. We label each party with the index $[i,j]$, as explained above, and we denote by $\tau_{[i,j]}$ a source distributing correlations to every party except party $[i,j]$. Moreover, we assume that the party $[L-1,0]$ can perform one of three measurements, $A_{[L-1,0];0}$, $A_{[L-1,0];1}$ or $A_{[L-1,0];2},$ while the rest of the parties $[i,j]$ can perform one of two measurements, $A_{[i,j];0}$ or $A_{[i,j];1}$.

The main method we use to construct the inequality capable of detecting LOSR-$\GMNL$ is the inflation technique. Given a network $\mathcal{O}$ described above, an inflation $\mathcal{I}$ of $\mathcal{O}$ is a network consisting of multiple copies of each party, multiple copies of each $(N-1)$-partite source of non-signalling correlations, and one global source of classical correlations. We assume that each copy of a given party is performing the same measurements as the original party in $\mathcal{O}$ and that each copy of a given source is distributing the same correlations as the original source $\tau_{[i,j]}$ in $\mathcal{O}$.

On top of the standard assumptions associated with the inflation technique, we also put additional assumptions on the inflations $\mathcal{I}$. First, we only consider inflations consisting of exactly two copies of each party $[i,j]$, which we denote by $[i,j]$ and $[i',j]$, and exactly two copies of each source $\tau_{[i,j]}$, which we denote by $\tau_{[i,j]}$ and $\tau_{[i',j]}$. Next, we require that party $[i,j]$ receives correlations from either $\tau_{[k,l]}$ or $\tau_{[k',l]}$ for all $[k,l]\neq [i,j]$. Together, these two assumptions imply that if party $[i,j]$ receives correlations from $\tau_{[k,l]}$ ($\tau_{[k',l]}$) then $[i',j]$ receives correlations from $\tau_{[k',l]}$ ($\tau_{[k,l]}$).

Let us consider a set of parties in $\mathcal{I}$ to which the source $\tau_{[i,j]}$ distributes correlations. Notice, that, using the above assumptions about $\mathcal{I}$, we can uniquely identify a source $\tau_{[i,j]}$ just by the set of non-primed parties that receive correlations from this source. Leveraging this fact, we introduce the notation $\tau_{[i,j]}=S$, where $S$ is a set of non-primed parties that receive correlations from $\tau_{[i,j]}$. From the symmetry of $\mathcal{I}$ it follows that $\tau_{[i,j]}\cup \tau_{[i',j]}=P\setminus \{[i,j]\}$, where $P$ is the set of all non-primed parties, which implies that the set $\{\tau_{[i,j]}\}_{[i,j]\in P}$ uniquely identifies $\mathcal{I}$.

To make the description of $\tau_{[i,j]}$ and other sets more concise, let us define the following sets:
\begin{align}\label{eq:set_short}
\begin{split}
(n,m)_{1}&=\{n+i;i\in\{0,1,\dots,m-n\}\}\,,\\
(n,m)_{2}&=\{n+2i;i\in\{0,1,\dots,(m-n)/2\}\}\,,\\
\end{split}
\end{align}
where we take $(n,m)_{1}=\emptyset=(n,m)_{2}$ if $n>m$. 

Lastly, we are going to use the following Lemma, which follows directly as a consequence of the results in \cite{Navascu_s_2020}.
\begin{lemma}\label{lem:ineq}
Let $M_{1}, M_{2}$ and $M_{3}$ be observables with outcomes $\pm 1$. Then the following inequality holds true:
\begin{equation}
\langle M_{1} M_{2} \rangle \geqslant \langle M_{1} M_{3} \rangle + \langle M_{2} M_{3} \rangle -1\,.
\end{equation}
\end{lemma}

With this, we can formulate the main result of this section.
\begin{theorem}
Let us consider an $N$-partite network with $(N-1)$-partite sources of non-signalling correlations and an $N$-partite source of classical correlations. Given that the party $[L-1,0]$ can perform one of three measurements while the rest of the parties have access to two measurements each and that $L\geqslant 3$, the resulting measurement probabilities are constrained by
\begin{align}\label{ineq:caterpillar}
\begin{split}
I_{L-2,L-1,L} +2 \sum_{k\in (0,L-3)_{1}} \Bigg(\left\langle A_{[k,0];0}A_{[k+1,0];1}A_{[k+2,0];0}\prod_{j=1}^{n_{k+1}}A_{[k+1,j];0}\right\rangle &+ 2\sum_{l=1}^{n_{k+2}}\left\langle A_{[k+2,0];0}A_{[k+2,l];1}\right\rangle \Bigg)\\ 
&+ 2\left\langle A_{[L-1,0];0}A_{[L,0];1}\right\rangle \leqslant 2(2N-L)\,,
\end{split}
\end{align}
where
\begin{align}\label{eq:I_210}
\begin{split}
I_{L-2,L-1,L} &= \left\langle A_{[L-1,0];1}A_{[L,0];1}\right\rangle  +\left\langle A_{[L-2,0];0}A_{[L-1,0];1}A_{[L,0];0}\prod_{j=1}^{n_{L-1}}A_{[L-1,j];0}\right\rangle\\ &-\left\langle A_{[L-1,0];2}A_{[L,0];1}\right\rangle + \left\langle A_{[L-2,0];0}A_{[L-1,0];2}A_{[L,0];0}\prod_{j=1}^{n_{L-1}}A_{[L-1,j];0}\right\rangle\,.
\end{split}
\end{align}
\end{theorem}
We divide the following proof into two parts. First, we give a complete proof for even $L$, then we give a proof for odd $L$ where we skip some steps if the derivation is the same for even and odd $L$.
\begin{proof}

\textit{Part 1: even $L$}

Let us begin by considering a network inflation $\mathcal{I}_{0}$ defined as
\begin{equation}\label{eq:inflation_0}
\mathcal{I}_{0}:\quad \tau_{[i,j]}=\begin{cases}
[(1,i-2)_{2}\cup (i,L)_{1},\cdot ]\setminus \{[i,j]\} &\textrm{for odd }i\textrm{ and all }j\,,\\
[(2,i-2)_{2}\cup (i,L)_{1},\cdot]\setminus \{[i,j]\}&\textrm{for even }i\textrm{ and all }j\,,
\end{cases}
\end{equation}
where for an arbitrary set $S$ we have $[S,\cdot]=\{[l,r]\}_{l\in S,r\in (0,n_{l})_{1}}$. Using \cite[Theorem 1]{Augusiak_2014}, we can write a monogamy relation that holds true in $\mathcal{I}_{0}$
\begin{equation}\label{eq:It}
I_{L-1,L;t}|_{\mathbf{a}_{e}=(-1)^{t}}+ 2\left\langle A_{[1',0];0}\prod_{i\in (2,L)_{2}}A_{[i,0];1}\prod_{j=1}^{n_{i}} A_{[i,j];0}\right\rangle \Big|_{\mathbf{a}_{e}=(-1)^{t}} \leqslant 4\,, 
\end{equation}
where
\begin{equation}
\mathbf{a}_{e}=\left(\prod_{i\in (1,L-3)_{2}}a_{[i,0];1}\prod_{j=1}^{n_{i}}a_{[i,j];0}\right) \prod_{j=1}^{n_{L-1}}a_{[L-1,j];0}\,,
\end{equation}
where $a_{[i,j];k}$ is a measurement result of $A_{[i,j];k}$, and
\begin{equation}\label{eqn: 2_Bell}
I_{L-1,L;t}=\langle A_{[L-1,0];1}A_{[L,0];1}\rangle+ (-1)^{t}\langle A_{[L-1,0];1}A_{[L,0];0}\rangle- \langle A_{[L-1,0];2}A_{[L,0];1}\rangle+(-1)^{t} \langle A_{[L-1,0];2}A_{[L,0];0}\rangle\,.
\end{equation}

We multiply both sides of \eqref{eq:It} by $P(\mathbf{a}_{e}=(-1)^{t})$ and we sum over $t\in\{0,1\}$
\begin{equation}\label{eq:ineq_av}
I_{L-1,L;\textrm{av}}+2 \left\langle A_{[1',0];0}\prod_{i\in (2,L)_{2}}A_{[i,0];1}\prod_{j=1}^{n_{i}} A_{[i,j];0}\right\rangle \leqslant 4\,,
\end{equation}
where
\begin{equation}
I_{L-1,L;\textrm{av}}=\sum_{t=0}^{1} P(\mathbf{a}_{e}=(-1)^{t})I_{L-1,L;t}|_{\mathbf{a}_{e}=(-1)^{t}}\,.
\end{equation}

Our goal is to formulate an inequality similar to Ineq. \eqref{eq:ineq_av} but for the original network $\mathcal{O}$ rather than an inflation $\mathcal{I}_{0}$. To this end, we have to bound the value of both terms on the left-hand side of \eqref{eq:ineq_av} by some other expected values calculated over correlations from $\mathcal{O}$. 

First, let us focus on $I_{L-1,L;\textrm{av}}$. Using the fact that 
\begin{equation}
\sum_{t=0}^{1}P(\mathbf{a}_{e}=(-1)^{t})\left\langle A_{[L-1,0];j}A_{[L,0];1}\right\rangle |_{\mathbf{a}_{e}=(-1)^{t}}=\left\langle A_{[L-1,0];j}A_{[L,0];1}\right\rangle
\end{equation}
for $j\in\{1,2\}$, we can rewrite $I_{L-1,L;\textrm{av}}$ as
\begin{align}\label{eq:I_av_sum}
\begin{split}
I_{L-1,L;\textrm{av}}&= \sum_{t=0}^{1}(-1)^{t}P(\mathbf{a}_{e}=(-1)^{t})\left\langle A_{[L-1,0];1}A_{[L,0];0}\right\rangle |_{\mathbf{a}_{e}=(-1)^{t}} +\sum_{t=0}^{1}(-1)^{t}P(\mathbf{a}_{e}=(-1)^{t})\left\langle A_{[L-1,0];2}A_{[L,0];0}\right\rangle |_{\mathbf{a}_{e}=(-1)^{t}}\\
&+\left\langle A_{[L-1,0];1}A_{[L,0];1}\right\rangle - \left\langle A_{[L-1,0];2}A_{[L,0];1}\right\rangle\,.
\end{split}
\end{align}
Let us focus on the first sum. We can expand the expected value in terms of probabilities which yields
\begin{align}
\begin{split}
\sum_{t=0}^{1}(-1)^{t}P(\mathbf{a}_{e}=(-1)^{t})\left\langle A_{[L-1,0];1}A_{[L,0];0}\right\rangle |_{\mathbf{a}_{e}=(-1)^{t}} &= \sum_{t,k=0}^{1}(-1)^{t+k}P(\mathbf{a}_{e}=(-1)^{t})P(a_{[L-1,0];1}a_{[L,0];0}=(-1)^{k}|\mathbf{a}_{e}=(-1)^{t})\\
&= \sum_{t,k=0}^{1}(-1)^{t+k}P(a_{[L-1,0];1}a_{[L,0];0}=(-1)^{k},\mathbf{a}_{e}=(-1)^{t})\\
&= \sum_{l=0}^{1}(-1)^{l}P(\mathbf{a}_{e}a_{[L-1,0];1}a_{[L,0];0}=(-1)^{l})= \left\langle \mathbf{A}_{e}A_{[L-1,0];1}A_{[L,0];0} \right\rangle\,,
\end{split}
\end{align}
where
\begin{equation}
\mathbf{A}_{e}=\left(\prod_{i\in (1,L-3)_{2}}A_{[i,0];1}\prod_{j=1}^{n_{i}}A_{[i,j];0}\right) \prod_{j=1}^{n_{L-1}}A_{[L-1,j];0}\,.
\end{equation}
Applying the same technique for the second sum in \eqref{eq:I_av_sum} allows us to conclude that
\begin{equation}\label{eq:pI}
I_{L-1,L;\textrm{av}}= \left\langle A_{[L-1,0];1}A_{[L,0];1}\right\rangle  +\left\langle \mathbf{A}_{e}A_{[L-1,0];1}A_{[L,0];0}\right\rangle -\left\langle A_{[L-1,0];2}A_{[L,0];1}\right\rangle + \left\langle \mathbf{A}_{e}A_{[L-1,0];2}A_{[L,0];0}\right\rangle\,. 
\end{equation}

Notice that this expression depends on measurements on parties $[(1, L-1)_{2},\cdot]$ and on a party $[L,0]$. Since a subnetwork consisting of those parties in $\mathcal{I}_{0}$ has the same structure as a subnetwork consisting of the same parties in $\mathcal{I}_{1}$, defined as
\begin{equation}\label{eq:inflation_1}
\mathcal{I}_{1}:\quad \tau_{[i,j]}=\begin{cases}
[(1,L)_{1},\cdot] \setminus \{[i,j]\} &\textrm{for odd }i \textrm{ and all }j\,,\\
[(i,L)_{1},\cdot]\setminus \{[i,j]\} &\textrm{for even }i \textrm{ and all }j\,,
\end{cases}
\end{equation}
we conclude that the value of \eqref{eq:pI} calculated over correlations from $\mathcal{I}_{0}$ is equal to the value of \eqref{eq:pI} calculated over correlations from $\mathcal{I}_{1}$. 

While we cannot relate the value of \eqref{eq:pI} calculated over $\mathcal{I}_{1}$ directly to the value of \eqref{eq:pI} calculated over $\mathcal{O}$, since the structure of a relevant subnetwork is different for $\mathcal{I}_{1}$ and $\mathcal{O}$, we can achieve it by performing an intermediate step. The plan is to bound the value of \eqref{eq:pI} from below by a function of other expected values in $\mathcal{I}_{1}$, which then can be related to the same expected values in $\mathcal{O}$. To this end, we make use of Lemma \ref{lem:ineq} to formulate the following inequality:
\begin{align}\label{eq:pI_bound}
\begin{split}
I_{L-1,L;\textrm{av}} &= \left\langle A_{[L-1,0];1}A_{[L,0];1}\right\rangle  +\left\langle \mathbf{A}_{e}A_{[L-1,0];1}A_{[L,0];0}\right\rangle -\left\langle A_{[L-1,0];2}A_{[L,0];1}\right\rangle + \left\langle \mathbf{A}_{e}A_{[L-1,0];2}A_{[L,0];0}\right\rangle \\
&\geqslant \left\langle A_{[L-1,0];1}A_{[L,0];1}\right\rangle  +\left\langle T_{0}\mathbf{A}_{e}A_{[L-1,0];1}A_{[L,0];0}\right\rangle +\left\langle T_{0}\right\rangle -1 \\
&-\left\langle A_{[L-1,0];2}A_{[L,0];1}\right\rangle + \left\langle T_{0} \mathbf{A}_{e}A_{[L-1,0];2}A_{[L,0];0}\right\rangle+\left\langle T_{0}\right\rangle -1 \\
&\geqslant \left\langle A_{[L-1,0];1}A_{[L,0];1}\right\rangle  +\left\langle A_{[L-2,0];0}A_{[L-1,0];1}A_{[L,0];0}\prod_{j=1}^{n_{L-1}}A_{[L-1,j];0}\right\rangle\\ &-\left\langle A_{[L-1,0];2}A_{[L,0];1}\right\rangle + \left\langle A_{[L-2,0];0}A_{[L-1,0];2}A_{[L,0];0}\prod_{j=1}^{n_{L-1}}A_{[L-1,j];0}\right\rangle +2\sum_{k\in(0,L-4)_{2}} (\left\langle T_{k} \right\rangle -1)\\
&= I_{L-2,L-1,L}+2\sum_{k\in(0,L-4)_{2}}\left\langle T_{k} \right\rangle -L + 2\,,
\end{split}
\end{align}
where $I_{L-2,L-1,L}$ is defined in \eqref{eq:I_210}, 
\begin{equation}\label{eq:T_j}
T_{k}=A_{[k,0];0}A_{[k+1,0];1}A_{[k+2,0];0}\prod_{j=1}^{n_{k+1}}A_{[k+1,j];0}\,,
\end{equation}
for any $k\in (0,L-1)_{1}$ and $A_{[0,0];0}=A_{[L+1,0];0}=\mathbb{1}$, and we used the fact that
\begin{equation}
\mathbf{A}_{e}=A_{[L-2,0];0}\prod_{j=1}^{n_{L-1}}A_{[L-1,j],0}\prod_{k\in(0,L-4)_{2}}T_{k}\,.
\end{equation}

Let us now consider the original network $\mathcal{O}$
\begin{equation}\label{eq:original}
\mathcal{O}:\quad \tau_{[i,j]}=[(1,L)_{1},\cdot]\setminus\{[i,j]\}\,. 
\end{equation}
Notice that the subnetwork consisting of parties $[L-2,0]$, $[L-1,\cdot]$ and $[L,0]$ has the same structure in both $\mathcal{I}_{1}$ and $\mathcal{O}$, which implies that the value of $I_{L-2,L-1,L}$ is the same in both of these networks. As for the terms $\left\langle T_{k}\right\rangle$, they cannot be related to $\mathcal{O}$ directly, as the subnetworks on which the $T_{k}$ terms act nontrivially are different in $\mathcal{I}_{1}$ and in $\mathcal{O}$. However, this problem can be averted by using a set of inflation tailored individually for each $T_{k}$.

Let us consider the following set of inflations
\begin{equation}\label{eq:inf_j}
\mathcal{J}_{k}^{m}:\quad \tau_{[i,j]}=\begin{cases}
[(i,L)_{1},\cdot]\setminus\{[i,j]\} &\textrm{for }i = k+2 \textrm{ and all }j \geqslant m+ 1\,,\\
[(1,L)_{1},\cdot]\setminus\{[i,j]\} &\textrm{else}\,.\\
\end{cases}
\end{equation}
Based on the equivalence of the relevant subnetworks, we have that, for $k\in(0,L-4)_{2}$, expected values $\left\langle T_{k}\right\rangle$ calculated over a state from $\mathcal{I}_{1}$ and from $\mathcal{J}_{k}^{0}$ are equal. We can now use Lemma \ref{lem:ineq} to get the following inequality:
\begin{align}
\begin{split}
\left\langle T_{k}\right\rangle \geqslant \left\langle A_{[k,0];0}A_{[k+1,0];1}A_{[k+2,1];1}\prod_{j=1}^{n_{k+1}}A_{[k+1,j];0} \right\rangle
+ \left\langle A_{[k+2,0];0} A_{[k+2,1];1}\right\rangle -1\,.
\end{split}
\end{align}
Since the subnetworks $\{[k+2,0],[k+2,1]\}$ in $\mathcal{J}_{k}^{0}$ and $\mathcal{O}$ are equivalent, we get that the second expected value on the right-hand side of the inequality calculated over correlations from $\mathcal{J}_{k}^{0}$ and $\mathcal{O}$ yields the same value. As for the first expected value on the right-hand side, we can conclude that the expected value calculated over correlations from $\mathcal{J}_{k}^{0}$ equals the same expected value calculated over correlations from $\mathcal{J}_{k}^{1}$.

We can again use Lemma \ref{lem:ineq} to bound this expected value by
\begin{equation}\label{eq:a_prod_dec}
    \left\langle A_{[k,0];0}A_{[k+1,0];1}A_{[k+2,1];1}\prod_{j=1}^{n_{k+1}}A_{[k+1,j];0} \right\rangle \geqslant \left\langle A_{[k,0];0}A_{[k+1,0];1}A_{[k+2,2];1}\prod_{j=1}^{n_{k+1}}A_{[k+1,j];0} \right\rangle + \left\langle A_{[k+2,1];1} A_{[k+2,2];1}\right\rangle -1\,.
\end{equation}
Again, we can relate the second expected value on the right to the same expected value calculated over correlations from $\mathcal{O}$, while relating the first expected value to the same expected value calculated over correlations from the inflation $\mathcal{J}_{k}^{2}$.

We repeat this procedure until we get to the following inequality,
\begin{align}
\begin{split}
\left\langle A_{[k,0];0}A_{[k+1,0];1}A_{[k+2,n_{k+2}-1];1}\prod_{j=1}^{n_{k+1}}A_{[k+1,j];0} \right\rangle &\geqslant \left\langle A_{[k,0];0}A_{[k+1,0];1}A_{[k+2,n_{k+2}];1}\prod_{j=1}^{n_{k+1}}A_{[k+1,j];0} \right\rangle \\
&+ \left\langle A_{[k+2,n_{k+2}-1];1} A_{[k+2,n_{k+2}];1}\right\rangle -1\,,   
\end{split}
\end{align}
where the expected values are calculated over correlations from $\mathcal{J}_{k}^{n_{k+2}-1}$. We can now use the same technique again to conclude that both expected values on the right side of the inequality are equal to the same expected values but calculated over correlations from $\mathcal{O}$.

In order for the final inequality to be well defined for any value of $n_{k+2}$ we have to use Lemma \ref{lem:ineq} twice more:
\begin{align}
\begin{split}
\left\langle A_{[k,0];0}A_{[k+1,0];1}A_{[k+2,n_{k+2}];1}\prod_{j=1}^{n_{k+1}}A_{[k+1,j];0} \right\rangle &\geqslant \left\langle A_{[k,0];0}A_{[k+1,0];1}A_{[k+2,0];0}\prod_{j=1}^{n_{k+1}}A_{[k+1,j];0} \right\rangle \\
&+ \left\langle A_{[k+2,0];0} A_{[k+2,n_{k+2}];1} \right\rangle -1\,,   
\end{split}
\end{align}
and
\begin{equation}
\left\langle A_{[k+2,l];1}A_{[k+2,l+1];1}\right\rangle \geqslant \left\langle A_{[k+2,l];1}A_{[k+2,0];0}\right\rangle + \left\langle A_{[k+2,l+1];1}A_{[k+2,0];0}\right\rangle -1
\end{equation}
for all $l\in (1,n_{k+2}-1)_{1}$.

The above method can be applied for all $k\in (0,L-4)_{2}$ meaning that every expected value $\left\langle T_{k} \right\rangle$ calculated over correlations from $\mathcal{I}_{1}$ can be bounded from below by expected values calculated over correlations from $\mathcal{O}$. Putting it all together we get
\begin{align}\label{ineq:part_1}
\begin{split}
(I_{L-1,L,\textrm{av}})_{\mathcal{I}_{0}} &\geqslant I_{L-2,L-1,L} +2 \sum_{k\in (0,L-4)_{2}} \Bigg(\left\langle A_{[k,0];0}A_{[k+1,0];1}A_{[k+2,0];0}\prod_{j=1}^{n_{k+1}}A_{[k+1,j];0}\right\rangle \\
&+ 2\sum_{l=1}^{n_{k+2}}\left\langle A_{[k+2,l];1}A_{[k+2,0];0}\right\rangle  - 2 n_{k+2} \Bigg) - L + 2\,,
\end{split}
\end{align}
where by writing $(I_{L-1,L,\textrm{av}})_{\mathcal{I}_{0}}$, we emphasize that the left-hand side of the inequality is calculated on correlations from $\mathcal{I}_{0}$ while the right-hand side is calculated on correlations from $\mathcal{O}$. 

We now come back to Ineq. \eqref{eq:ineq_av} and we focus on the term $\langle A_{[1',0];0}\prod_{i\in (2,L)_{2}}A_{[i,0];1}\prod_{j=1}^{n_{i}} A_{[i,j];0}\rangle$ in the inflation $\mathcal{I}_0$. Let us consider the following inflation:
\begin{equation}\label{eq:inflation_2}
\mathcal{I}_{2}:\quad \tau_{[i,j]}=\begin{cases}
\left(\{[1,0]\}\cup[(i,L)_{1},\cdot] \right) \setminus\{[i,j]\}  &\textrm{for odd }i\textrm{ and all }j\,,\\
[(2,N)_{1},\cdot]\setminus\{[i,j]\}  &\textrm{for even }i\textrm{ and all }j\,.
\end{cases}
\end{equation}
Due to the equivalence of the relevant subnetworks in $\mathcal{I}_{0}$ and $\mathcal{I}_{2}$, we have that the expected value $\langle A_{[1',0];0}\prod_{i\in (2,L)_{2}}A_{[i,0];1}\prod_{j=1}^{n_{i}} A_{[i,j];0}\rangle$ calculated over correlations from $\mathcal{I}_{0}$ is equal to the same expected value calculated for the correlations from $\mathcal{I}_{2}$. 

Let us now consider
\begin{equation}\label{eq:inflation_3}
\mathcal{I}_{3}:\quad \tau_{[i,j]}=\begin{cases}
[(i,L)_{1},\cdot] \setminus\{[i,j]\}  &\textrm{for odd }i\textrm{ and all }j\,,\\
[(1,N)_{1},\cdot]\setminus\{[i,j]\}  &\textrm{for even }i\textrm{ and all }j\,.
\end{cases}
\end{equation}
This inflation is derived from $\mathcal{I}_{2}$ by simply swapping the labels of parties $1$ and $1'$. We therefore have that $\langle A_{[1',0];0}\prod_{i\in (2,L)_{2}}A_{[i,0];1}\prod_{j=1}^{n_{i}} A_{[i,j];0}\rangle$ calculated over correlations from $\mathcal{I}_{2}$ equals $\langle A_{[1,0];0}\prod_{i\in (2,L)_{2}}A_{[i,0];1}\prod_{j=1}^{n_{i}} A_{[i,j];0}\rangle$ (notice the lack of a prime in the first observable) calculated over correlations from $\mathcal{I}_{3}$. 

We can bound this expected value using Lemma \ref{lem:ineq} and the fact that
\begin{equation}
 A_{[1,0];0}\prod_{i\in (2,L)_{2}}A_{[i,0];1}\prod_{j=1}^{n_{i}} A_{[i,j];0} = \prod_{k\in (1,L-1)_{2}} T_{k}\,,
\end{equation}
which gives us
\begin{align}\label{eq:exp_bound}
\begin{split}
\left\langle A_{[1,0];0}\prod_{i\in (2,L)_{2}}A_{[i,0];1}\prod_{j=1}^{n_{i}} A_{[i,j];0}\right\rangle &\geqslant \left\langle A_{[1,0];0}\prod_{i\in (2,L-2)_{2}}A_{[i,0];1}\prod_{j=1}^{n_{i}} A_{[i,j];0} A_{[L-1,0];0}\right\rangle  + \left\langle T_{L-1}\right\rangle - 1\\
&\geqslant \sum_{k\in (1,L-1)_{2}}\left\langle T_{k}\right\rangle - \frac{L}{2} +1\,.
\end{split}
\end{align}

In order to bound the expected values $\left\langle T_{k}\right\rangle$ by correlations from $\mathcal{O}$, we once again make use of the inflations $\mathcal{J}_{k}^{m}$ in the exact same way as before, which results in the following inequality,
\begin{align}\label{ineq:part_2}
\begin{split}
\left\langle A_{[1',0];0}\prod_{i\in (2,L)_{2}}A_{[i,0];1}\prod_{j=1}^{n_{i}} A_{[i,j],0}\right\rangle_{\mathcal{I}_{0}} &\geqslant \sum_{k\in (1,L-1)_{2}}\Bigg(\left\langle A_{[k,0];0}A_{[k+1,0];1}A_{[k+2,0];0}\prod_{j=1}^{n_{k+1}}A_{[k+1,j];0}\right\rangle \\
& + 2\sum_{l=1}^{n_{k+2}}\left\langle A_{[k+2,l];1}A_{[k+2,0];0}\right\rangle - 2n_{k+2}  \Bigg) - \frac{L}{2} +1\,,      
\end{split}
\end{align}
where expected values on the right-hand side of the inequality are calculated over correlations from $\mathcal{O}$.

We can finally combine Ineq.s \eqref{eq:ineq_av}, \eqref{ineq:part_1} and \eqref{ineq:part_2} into a single inequality
\begin{align}
\begin{split}
I_{L-2,L-1,L} +2 \sum_{k\in (0,L-3)_{1}} \Bigg(\left\langle A_{[k,0];0}A_{[k+1,0];1}A_{[k+2,0];0}\prod_{j=1}^{n_{k+1}}A_{[k+1,j];0}\right\rangle &+ 2\sum_{l=1}^{n_{k+2}}\left\langle A_{[k+2,0];0}A_{[k+2,l];1}\right\rangle \Bigg)\\ 
&+ 2\left\langle A_{[L-1,0];0}A_{[L,0];1}\right\rangle \leqslant 2(2N-L)\,.
\end{split}
\end{align}

\textit{Part 2: odd $L$}

Let us begin by considering network inflation $\mathcal{I}_{0}$ defined in Eq. \eqref{eq:inflation_0}. Using \cite[Theorem 1]{Augusiak_2014}, we can write a monogamy relation that holds true in $\mathcal{I}_{0}$,
\begin{equation}\label{eq:It_odd}
I_{L-1,L;t}|_{\mathbf{a}_{o}=(-1)^{t}}+ 2\left\langle \prod_{i\in (1,L)_{2}}A_{[i,0];1}\prod_{j=1}^{n_{i}}A_{[i,j];0}\right\rangle |_{\mathbf{a}_{o}=(-1)^{t}} \leqslant 4\,, 
\end{equation}
where
\begin{equation}
\mathbf{a}_{o}=a_{[1',0];0}\left(\prod_{i\in (2,L-3)_{2}}a_{[i,0];1}\prod_{j=1}^{n_{i}}a_{[i,j];0}\right)\prod_{j=1}^{n_{L-1}}a_{[L-1,j];0}\,.
\end{equation}
We multiply both sides of \eqref{eq:It_odd} by $P(\mathbf{a}_{o}=(-1)^{t})$ and we sum over $t\in\{0,1\}$
\begin{equation}\label{eq:ineq_av_odd}
I_{L-1,L;\textrm{av}}+2 \left\langle \prod_{i\in (1,L)_{2}}A_{[i,0];1}\prod_{j=1}^{n_{i}}A_{[i,j];0}\right\rangle \leqslant 4\,,
\end{equation}
where
\begin{equation}
I_{L-1,L;\textrm{av}}=\sum_{t=0}^{1} P(\mathbf{a}_{o}=(-1)^{t})I_{L-1,L;t}|_{\mathbf{a}_{o}=(-1)^{t}}\,.
\end{equation}
Using the same technique as in the first part of the proof, we get
\begin{equation}
I_{L-1,L;\textrm{av}}= \left\langle A_{[L-1,0];1}A_{[L,0];1}\right\rangle  +\left\langle \mathbf{A}_{o}A_{[L-1,0];1}A_{[L,0];0}\right\rangle -\left\langle A_{[L-1,0];2}A_{[L,0];1}\right\rangle + \left\langle \mathbf{A}_{o}A_{[L-1,0];2}A_{[L,0];0}\right\rangle\,,
\end{equation}
where
\begin{equation}
\mathbf{A}_{o}=A_{[1',0];0}\left(\prod_{i\in (2,L-3)_{2}}A_{[i,0];1}\prod_{j=1}^{n_{i}}A_{[i,j];0}\right)\prod_{j=1}^{n_{L-1}}A_{[L-1,j];0}\,.
\end{equation}
Our goal is to formulate an inequality similar to Ineq. \eqref{eq:ineq_av_odd} but over correlations from the original network $\mathcal{O}$ rather than an inflation $\mathcal{I}_{0}$.

Let us first focus on $I_{L-1,L;\textrm{av}}$. Notice that this expression depends only on measurements on parties $[(2,L-1)_{2},\cdot]$, $[1',0]$, $[L,0]$. Since a subnetwork consisting of those parties in $\mathcal{I}_{0}$ is equivalent to a subnetwork consisting of the same parties in $\mathcal{I}_{2}$, defined by Eq. \eqref{eq:inflation_2}, we conclude that the value of $I_{L-1,L;\textrm{av}}$ calculated over correlations from $\mathcal{I}_{0}$ is equal to the value of $I_{L-1,L;\textrm{av}}$ calculated over correlations from $\mathcal{I}_{2}$. 

Moreover, $I_{L-1,L;\textrm{av}}$ calculated over correlations from $\mathcal{I}_{2}$ is equal to
\begin{equation}
I_{L-1,L;\tilde{\textrm{av}}}= \left\langle A_{[L-1,0];1}A_{[L,0];1}\right\rangle  +\left\langle \tilde{\mathbf{A}}_{o}A_{[L-1,0];1}A_{[L,0];0}\right\rangle -\left\langle A_{[L-1,0];2}A_{[L,0];1}\right\rangle + \left\langle \tilde{\mathbf{A}}_{o}A_{[L-1,0];2}A_{[L,0];0}\right\rangle 
\end{equation}
calculated over correlations from $\mathcal{I}_{3}$, where
\begin{equation}
\tilde{\mathbf{A}}_{o}=A_{[1,0];0}\left(\prod_{i\in (2,L-3)_{2}}A_{[i,0];1}\prod_{j=1}^{n_{i}}A_{[i,j];0}\right)\prod_{j=1}^{n_{L-1}}A_{[L-1,j];0}\,.
\end{equation}

As was the case for even $L$, here we can bound the value of $I_{L-1,L;\tilde{\textrm{av}}}$ from below by
\begin{align}
\begin{split}
 I_{L-1,L;\tilde{\textrm{av}}} \geqslant I_{L-2,L-1,L}+2\sum_{k\in(1,L-4)_{2}}\left\langle T_{k} \right\rangle -L + 3\,,
\end{split}
\end{align}
where $I_{N-2,N-1,N}$ is defined in \eqref{eq:I_210}.

Let us now consider the original network $\mathcal{O}$ \eqref{eq:original}. Since the relevant subnetwork for $I_{L-2,L-1,L}$ is equivalent in $\mathcal{I}_{3}$ and $\mathcal{O}$, the value of 
$I_{L-2,L-1,L}$ calculated over correlations from $\mathcal{I}_{3}$ and $\mathcal{O}$ is the same. As for $\left\langle T_{k} \right\rangle$, these expected values can be bounded from below by utilizing inflations $\mathcal{J}_{k}^{m}$ \eqref{eq:inf_j} for $m\in (0,n_{k+2}-1)$, which at the end gives us
\begin{align}\label{eq:pI_bound_odd}
\begin{split}
(I_{L-1,L,\textrm{av}})_{\mathcal{I}_{0}} &\geqslant I_{L-2,L-1,L} +2 \sum_{k\in (1,L-4)_{2}} \Bigg(\left\langle A_{[k,0];0}A_{[k+1,0];1}A_{[k+2,n_{k+2}];1}\prod_{j=1}^{n_{k+1}}A_{[k+1,j];0}\right\rangle \\
&+ 2\sum_{l=1}^{n_{k+2}}\left\langle A_{[k+2,l];1}A_{[k+2,0];0}\right\rangle - 2n_{k+2}  \Bigg) - L + 3\,,
\end{split}
\end{align}
where the left-hand side is calculated over correlations from $\mathcal{O}$.

We now come back to Ineq. \eqref{eq:ineq_av_odd} and we focus on the term $\langle \prod_{i\in (1,L)_{2}}A_{[i,0];1}\prod_{j=1}^{n_{i}}A_{[i,j];0}\rangle$. Notice that it depends only on parties $[(1,L)_{2},\cdot]$. Since $\mathcal{I}_{0}$ \eqref{eq:inflation_0} and $\mathcal{I}_{1}$ \eqref{eq:inflation_1} differ only by connections to parties $[(2,L-1)_{2},\cdot]$, we can conclude that 
\begin{equation}
\left\langle \prod_{i\in (1,L)_{2}}A_{[i,0];1}\prod_{j=1}^{n_{i}}A_{[i,j];0}\right\rangle_{\mathcal{I}_{0}}=\left\langle \prod_{i\in (1,L)_{2}}A_{[i,0];1}\prod_{j=1}^{n_{i}}A_{[i,j];0}\right\rangle_{\mathcal{I}_{1}}\,.
\end{equation}

We then use Lemma \ref{lem:ineq} to bound this expected value from below 
\begin{equation}
 \left\langle \prod_{i\in (1,L)_{2}}A_{[i,0];1}\prod_{j=1}^{n_{i}}A_{[i,j];0}\right\rangle\geqslant \sum_{k\in(0,L-1)_{2}}\left\langle T_{k}\right\rangle -\frac{L-1}{2}\,.
\end{equation}

Next, each $\left\langle T_{k}\right\rangle$ can itself be bounded from below using the inflations $\mathcal{J}_{k}^{m}$ \eqref{eq:inf_j}, which in the end results in
\begin{align}\label{eq:exp_bound_odd}
\begin{split}
 \left\langle \prod_{i\in (1,L)_{2}}A_{[i,0];1}\prod_{j=1}^{n_{i}}A_{[i,j];0}\right\rangle_{\mathcal{I}_{0}} &\geqslant \sum_{k\in(0,L-1)_{2}}\Bigg(\left\langle A_{[k,0];0}A_{[k+1,0];1}A_{[k+2,n_{k+2}];1}\prod_{j=1}^{n_{k+1}}A_{[k+1,j];0}\right\rangle \\
& + 2\sum_{l=1}^{n_{k+2}}\left\langle A_{[k+2,l];1}A_{[k+2,0];0}\right\rangle - 2n_{k+2} \Bigg) -\frac{L-1}{2}\,,
\end{split}
\end{align}
where the right-hand side is evaluated over correlations from $\mathcal{O}$.

Putting \eqref{eq:pI_bound_odd} and \eqref{eq:exp_bound_odd} together, we get the following inequality fulfilled by correlations from $\mathcal{O}$
\begin{align}
\begin{split}
I_{L-2,L-1,L} +2 \sum_{k\in (0,L-3)_{1}} \Bigg(\left\langle A_{[k,0];0}A_{[k+1,0];1}A_{[k+2,0];0}\prod_{j=1}^{n_{k+1}}A_{[k+1,j];0}\right\rangle &+ 2\sum_{l=1}^{n_{k+2}}\left\langle A_{[k+2,0];0}A_{[k+2,l];1}\right\rangle \Bigg)\\ 
&+ 2\left\langle A_{[L-1,0];0}A_{[L,0];1}\right\rangle \leqslant 2(2N-L)\,.
\end{split}
\end{align}

\end{proof}

After deriving an inequality detecting LOSR-$\GMNL$, we can now show that the caterpillar state $\ket{\ddagger_{N}}$, as well as any state $\rho$ sufficiently close to it, is LOSR-$\GMNL$.
\begin{corollary}\label{cor:noise}
Let us consider a state $\rho_{\ddagger_{N}}(\eta)$ defined as
\begin{equation}
\rho_{\ddagger_{N}}(\eta) = \eta \ket{\ddagger_{N}}\!\bra{\ddagger_{N}} + \frac{1-\eta}{2^{N}}\mathbb{1}\,.
\end{equation}
Given that $L\geqslant 3$, $\rho_{\ddagger_{N}}(\eta)$ is LOSR-$\GMNL$ for
\begin{equation}\label{eq:ineq_eta}
\eta > \frac{2N-L}{2N-L+\sqrt{2}-1}\,.
\end{equation}
\end{corollary}
\begin{proof}
We take the measurements as follows:
\begin{align}\label{eq:optimal_observables}
\begin{split}
A_{[i,j];0} = Z\,, \quad A_{[i,j];1} =X\qquad \textrm{for } [i,j]\in [(1,L)_{1},\cdot]\setminus\{[L-1,0]\}\,,\\
\quad A_{[L-1,0];0}=Z\,, \quad A_{[L-1,0];1}=\frac{1}{\sqrt{2}}(X+Z)\,,\quad A_{[L-1,0];2}=\frac{1}{\sqrt{2}}(X-Z)\,.
\end{split}
\end{align}
and then, using the stabilizing operators \eqref{eqn: caterpillar_stabs} of a caterpillar graph state, we calculate each term in Ineq. \eqref{ineq:caterpillar} for the state $\rho_{\ddagger_{N}}(\eta)$:
\begin{equation}
I_{L-2,L-1,L}=2\sqrt{2}\eta\,,\quad \langle \cdot \rangle= \eta\,,
\end{equation}
where by $\langle \cdot \rangle$ we mean that it holds true for all expected values in Ineq. \eqref{ineq:caterpillar}.
Substituting this into \eqref{ineq:caterpillar} yields Ineq. \eqref{eq:ineq_eta}. 
\end{proof}

\section{Cluster states are LOSR-GMNL for any prime local dimension}\label{app:C_N^d}
In this section, we prove that a cluster state for any prime local dimension is LOSR-$\GMNL$. We achieve this by utilizing the inflation technique to show that the assumption that the cluster state is not LOSR-$\GMNL$ violates a monogamy relation from \cite{Augusiak_2014}. Most of the details concerning the considered scenario (the inflation technique, etc.) are the same here as in Appendix \ref{app:C_N}, so we will not repeat them here. There are only three substantial changes that need to be highlighted: different party labels, different quantum states under examination, and a generalized monogamy relation suitable for subsystems with local dimension higher than two.

First, here we label all of the parties simply using a set of numbers $\{1,\dots,N\}$. Second, here we consider the cluster state $\ket{C_{N}^{d}}$ a state of prime local dimension $d$ which is stabilized by
\begin{align}\label{eq:generators}
\begin{split}
g_{1}&=X_{1}Z_{2}\,,\\
g_{i}&=Z_{i-1}X_{i}Z_{i+1}\quad \textrm{for } i\in(2,N-1)_{1}\,,\\
g_{N}&=Z_{N-1}X_{N}\,,
\end{split}
\end{align}
where $X$ and $Z$ are defined in \eqref{eq:def_xz} (for an overview of graph states see Appendix \ref{app:graph_states}).

Third, the aforementioned monogamy relation \cite[Theorem 1]{Augusiak_2014} is given by
\begin{equation}\label{eqn: monog_orig}
I_{m,n}+d(1-P(a_{n;i}=a_{l;k}))\geqslant d-1\,,
\end{equation}
where $a_{l;k}\in\{\omega^{j}\}_{j=0}^{d-1}$ denotes the measurement outcome of measurement $k$ of any party $l\neq m,n$ and 
\begin{equation}\label{eq:cglmp_omega}
I_{m,n}= \sum_{k=0}^{d-1}\sum_{i=1}^{2}k\Big( P(a_{m;i}=\omega^{k}a_{n;i}) + P(a_{n;i}=\omega^{k+\delta_{i,1}}a_{m;i+1}) \Big)\geqslant d-1\,
\end{equation}
is the Collins, Gisin, Linden, Massar, Popescu (CGLMP) inequality \cite{PhysRevLett.88.040404},  where $a_{m;3} \equiv a_{m;1}$. Measurements that lead to a violation of CGLMP inequality $I_{m, n}$ by a maximally entangled state $\ket{\phi^{+}}:=\frac{1}{\sqrt{d}}\sum_{j}\ket{jj}$ \cite{PhysRevLett.88.040404} can be presented as the following observables 
\begin{equation}\label{eq:observables_phi}
A_{m;j}=U_{\alpha_{j}}X^{\dagger}U_{\alpha_{j}}^{\dagger} \,,\qquad A_{n;j}=U_{\beta_{j}}XU_{\beta_{j}}^{\dagger}\,,
\end{equation}
where $U_{\chi}=\sum_{j}\omega^{j\chi}\ket{j}\!\bra{j}$ for any $\chi$, and $\alpha_{1}=0$, $\alpha_{2}=1/2$, $\beta_{1}=1/4$, $\beta_{2}=-1/4$.

For the purposes of our proof, we need to slightly adjust these observables. First, we need the first measurement of party $n$ to be equal to $X$. To this end, we make use of the property $V^{*}\otimes V\ket{\phi^{+}}=\ket{\phi^{+}}$ that holds for all unitary matrices $V$. Taking $V=U_{\beta_{1}}$ gives us a set of observables that also violate CGLMP with $\ket{\phi^{+}}$. In terms of Eq. \eqref{eq:observables_phi} the new observables are described by the following set of coefficients $\alpha_{1}=1/4$, $\alpha_{2}=3/4$, $\beta_{1}=0$, $\beta_{2}=-1/2$ . 

The last adjustment that we need to make is that we need a violation of CGLMP \eqref{eq:cglmp_omega} for the state $\ket{C_{2}^{d}}$ rather than $\ket{\phi^{+}}$. We use the relation $\ket{C_{2}^{d}}=F\otimes \mathbb{1}  \ket{\phi^{+}}$, where $F$ is the discrete Fourier transform matrix
\begin{equation}
F=\sum_{i,j=0}^{d-1}\omega^{ij}\ket{i}\!\bra{j}\,,
\end{equation}
to find that the state $\ket{C_{2}^{d}}$ violates CGLMP inequality \eqref{eq:cglmp_omega} with observables
\begin{equation}
A_{m;j}=FU_{\alpha_{j}}X^{\dagger}U_{\alpha_{j}}^{\dagger}F^{\dagger} \,,\qquad A_{n;j}=U_{\beta_{j}}XU_{\beta_{j}}^{\dagger}\,,
\end{equation}
with coefficients $\alpha_{1}=1/4$, $\alpha_{2}=3/4$, $\beta_{1}=0$ and $\beta_{2}=-1/2$ . 

With that, we can proceed to the main result of this section.

\begin{theorem}
A cluster state $\ket{C_{N}^{d}}$ is LOSR-$\GMNL$ for all prime $d$ and $N\geqslant 3$.
\end{theorem}
As was the case in Appendix \ref{app:C_N}, here we also split the proof into two parts. We first show a complete proof for even $N$, after which we show a proof for odd $N$ in which we omit derivations common for the cases of even and odd $N$.
\begin{proof}

\textit{Part 1: even $N$}

We begin by describing correlations that can be produced by the cluster $\ket{C_{N}^{d}}$, and we then prove by contradiction that these correlations cannot be produced in the network.

From Eq. \eqref{eq:generators}, it follows that, from the projection of $\ket{C_{N}^{d}}$ onto the state $\ket{0}_{N-2}$, we get a state $\ket{\psi}$ fulfilling
\begin{equation}
\ket{\psi}= \eta \ket{0}_{N-2}\!\bra{0}\ \ket{C_{N}^{d}}= \eta \ket{0}_{N-2}\!\bra{0}\ g_{i}\ket{C_{N}^{d}} = \eta g_{i}\ket{0}_{N-2}\!\bra{0}\ \ket{C_{N}^{d}}= g_{i} \ket{\psi}\,
\end{equation}
for all $i\neq N-2$, where $\eta$ is the normalization constant. From this relation it follows that $\ket{\psi}$ is stabilized by 
\begin{align}\label{eq:generators_n-2}
\begin{split}
g_{1}&=X_{1}Z_{2}\,,\qquad g_{i}=Z_{i-1}X_{i}Z_{i+1}\quad \textrm{for } i\in(2,N-4)_{1}\,,\\
g_{N-3}&=Z_{N-4}X_{N-3}\,,\qquad g_{N-2}=Z_{N-2}\,,\\
g_{N-1}&=X_{N-1}Z_{N}\,,\qquad g_{N}=Z_{N-1}X_{N}\,.
\end{split}
\end{align}
Therefore, we have that
\begin{equation}
\ket{\psi}= \ket{C_{N-3}}_{1,\dots,N-3}\otimes\ket{0}_{N-2}\otimes \ket{C_{2}^{d}}\,.
\end{equation}
Let us choose the following observables:
\begin{equation}\label{eq:observables_qudit}
A_{i;j}=X^{j}Z^{1-j}\quad \textrm{for }i\in(1,N-2)_{1}\,,\qquad A_{N-1;j}=FU^{\dagger}_{\beta_{j}}X^{\dagger}U_{\beta_{j}}F^{\dagger}\,,\qquad A_{N-1;0}=Z\,,\qquad A_{N;j}=U_{\alpha_{j}}^{\dagger}XU_{\alpha_{j}}\,. 
\end{equation}
Under this choice of observables, we have
\begin{equation}\label{eqn: CGLMP_violation}
I_{N-1,N}|_{a_{N-2;0}=1}<d-1\,,
\end{equation}
and, from the stabilizing operators \eqref{eq:generators} of $\ket{C_{N}^{d}}$, we can infer the expected values
\begin{align}
    \left\langle A_{i;0}  A_{i+1;1}  A_{i+2;0} \right\rangle =1 \qquad \forall \qquad i\in(0,N-3)_{1}\,,
\end{align}
which imply that
\begin{equation}\label{eqn: all_stabs_qudit}
P(a_{i;0}a_{i+1;1}a_{i+2;0}=1)=1 \qquad \forall \qquad  i\in(0,N-3)_{1}\,,
\end{equation}
where $a_{i;j}$ is the measurement result of measurement $A_{i;j}$ and we take $a_{0;0}:=1$. We also have $\left\langle A_{N-1;0} A_{N;1}\right\rangle =1$, and therefore 
\begin{align}\label{eqn: last_two}
    P(a_{N-1;0} a_{N;1}=1)=1\,.
\end{align}
We will now show that the \eqref{eqn: CGLMP_violation}, \eqref{eqn: all_stabs_qudit} and \eqref{eqn: last_two} cannot be satisfied simultaneously by any correlations originating from a network $\mathcal{O}$ consisting of $N$ parties connected with $(N-1)$-partite sources of non-signalling correlations and an $N$-partite source of shared randomness. We will achieve this by assuming the contrary, i.e. that \eqref{eqn: CGLMP_violation}, \eqref{eqn: all_stabs_qudit} and \eqref{eqn: last_two} hold true for some correlation from $\mathcal{O}$, which will lead us to a contradiction.

More specifically, we consider a network inflation
\begin{equation}\label{eq:inflation_0_d}
\mathcal{I}_{0}:\quad \tau_{i}=\begin{cases}
(1,i-2)_{2}\cup (i,L)_{1}\setminus \{i\} &\textrm{for odd }i\,,\\
(2,i-2)_{2}\cup (i,L)_{1}\setminus \{i\}&\textrm{for even }i\,,
\end{cases}
\end{equation}
in which we will show that the above set of assumptions leads to the violation of the following monogamy relation \eqref{eqn: monog_orig}:
\begin{equation}\label{eq:monogamy_qudit}
I_{N-1,N}|_{a_{1;1}^{\xi_{1}}a_{3;1}^{\xi_{3}}\dots a_{N-3;1}^{\xi_{N-3}}=1}+d(1-P(a_{1';0}^{\chi_{1}}a_{2;1}^{\chi_{2}}a_{4;1}^{\chi_{4}}\dots,a_{N;1}^{\chi_{N}}=1|a_{1;1}^{\xi_{1}}a_{3;1}^{\xi_{3}}\dots a_{N-3;1}^{\xi_{N-3}}=1))\geqslant d-1\,,
\end{equation}
where $I_{N-1,N}$ is the CGLMP inequality defined in \eqref{eq:cglmp_omega},
\begin{equation}\label{eq:chi_def}
\chi_{n}=\begin{cases}
1 &\textrm{for }n=1,2\mod 4\,,\\
d-1 &\textrm{for }n=0,3\mod 4\,,
\end{cases}
\end{equation}
and $\xi_{j}=- \chi_{j}\chi_{N-3}^{-1}\mod d$, with $\chi_{N-3}^{-1}$ the multiplicative inverse of $\chi_{N-3}$ in $\mathbb{Z}_{d}$ (since $d$ is prime, this inverse is guaranteed to exist).

We begin by considering what we can infer about the value of $I_{N-1,N}|_{a_{1;1}^{\xi_{1}}a_{3;1}^{\xi_{3}}\dots a_{N-3;1}^{\xi_{N-3}}=1}$ from the assumptions \eqref{eqn: CGLMP_violation}, \eqref{eqn: all_stabs_qudit} and \eqref{eqn: last_two}. To this end, let us consider the original network $\mathcal{O}$ and its inflation $\mathcal{I}_{1}$
\begin{equation}\label{eq:inflation_1_d}
\mathcal{I}_{1}:\quad \tau_{i}=\begin{cases}
(1,L)_{1} \setminus \{i\} &\textrm{for odd }i\,,\\
(i,L)_{1}\setminus \{i\} &\textrm{for even }i\,.
\end{cases}
\end{equation}
Notice that the subnetworks consisting of three consecutive parties $i,i+1,i+2$ for $i\in(0,N-2)_{2}$  (we ignore party $0$) have the same structure in $\mathcal{O}$ and in $\mathcal{I}_{1}$ \eqref{eq:inflation_1_d}. Therefore, by the assumption that \eqref{eqn: CGLMP_violation}, \eqref{eqn: all_stabs_qudit} and \eqref{eqn: last_two} are true in $\mathcal{O}$, we conclude that conditions
\begin{equation}
I_{N-1,N}|_{a_{N-2;0}=0}<d-1\,,\qquad a_{i;0}a_{i+1;1}a_{i+2;0}=1 \qquad \forall\,\,\, i\in(0,N-2)_{2}
\end{equation}
hold in $\mathcal{I}_{1}$. It follows then, that
\begin{equation}
(a_{1;1}a_{2;0})^{\chi_{1}}(a_{2;0}a_{3;1}a_{4;0})^{\chi_{3}}\dots (a_{N-4;0}a_{N-3;1}a_{N-2;0})^{\chi_{N-3}}=1
\end{equation}
where $\chi_{i}$ are defined in Eq. \eqref{eq:chi_def}, and the equality follows from the fact that each term in parentheses equals $1$. Using Eq. \eqref{eq:chi_def}, we can then conclude that
\begin{equation}
(a_{1;1}a_{2;0})^{\chi_{1}}(a_{2;0}a_{3;1}a_{4;0})^{\chi_{3}}\dots (a_{N-4;0}a_{N-3;1}a_{N-2;0})^{\chi_{N-3}}=a_{1;1}^{\chi_{1}}a_{3;1}^{\chi_{3}}\dots a_{N-3;1}^{\chi_{N-3}}a_{N-2;0}^{\chi_{N-3}}=1\,,
\end{equation}
which in turn implies:
\begin{equation}
a_{N-2;0}=\prod_{i\in (1,N-3)_{2}}a_{i;1}^{\xi_{i}}\,,
\end{equation}
where, as a reminder, $\xi_{i}=-\chi_{i} \chi^{-1}_{N-3} \mod d$. Substituting for $a_{N-2;0}$ in $I_{N-1,N}|_{a_{N-2;0}=1}$ gives us
\begin{equation}\label{eq:cglmp_I1}
I_{N-1,N}|_{a_{1;1}^{\xi_{1}}a_{3;1}^{\xi_{3}}\dots a_{N-3;1}^{\xi_{N-3}}=1}<d-1\,.
\end{equation}
Notice that this expression depends only on the measurement results of odd-numbered parties and the party $N$. The subnetwork consisting of the relevant parties is equivalent in $\mathcal{I}_{1}$ and $\mathcal{I}_{0}$ \eqref{eq:inflation_0_d}, so we conclude that \eqref{eq:cglmp_I1} holds true for $\mathcal{I}_{0}$ if \eqref{eqn: CGLMP_violation}, \eqref{eqn: all_stabs_qudit} and \eqref{eqn: last_two} hold in the original network. 

Coming back to the monogamy relation \eqref{eq:monogamy_qudit}, we can now calculate the value of the term 
\begin{equation}
P(a_{1';0}^{\chi_{1}}a_{2;1}^{\chi_{2}}a_{4;1}^{\chi_{4}}\dots,a_{N;1}^{\chi_{N}}=1|a_{1;1}^{\xi_{1}}a_{3;1}^{\xi_{3}}\dots a_{N-3;1}^{\xi_{N-3}}=1)\,.  
\end{equation}
We start again from $\mathcal{O}$ and use the consequences of our assumption that \eqref{eqn: CGLMP_violation}, \eqref{eqn: all_stabs_qudit} and \eqref{eqn: last_two} can be fulfilled by a correlation originating from $\mathcal{O}$. In particular, this means that
\begin{equation}\label{eq:three_a_odd}
a_{i;0}a_{i+1;1}a_{i+2;0}=1\,, \quad \textrm{for all } i\in(1,\dots,N-3)_{2},\qquad a_{N-1;0}a_{N;1}=1\,
\end{equation}
hold true in $\mathcal{O}$. Since the subnetwork containing any three parties $i,i+1,i+2$ for $i\in(1,\dots,N-1)_{2}$ (we ignore the party $N+1$) has the same structure in $\mathcal{O}$ as in $\mathcal{I}_{3}$, where 
\begin{equation}\label{eq:inflation_3_d}
\mathcal{I}_{3}:\quad \tau_{i}=\begin{cases}
(i,L)_{1} \setminus\{i\}  &\textrm{for odd }i\,,\\
(1,N)_{1}\setminus\{i\}  &\textrm{for even }i\,,
\end{cases}
\end{equation}
we conclude that \eqref{eq:three_a_odd} also holds true in $\mathcal{I}_{3}$. Furthermore, since $\mathcal{I}_{2}$ 
\begin{equation}\label{eq:inflation_2_d}
\mathcal{I}_{2}:\quad \tau_{i}=\begin{cases}
\{1\}\cup (i,L)_{1}  \setminus\{i\}  &\textrm{for odd }i\,,\\
(2,N)_{1}\setminus\{i\}  &\textrm{for even }i\,,
\end{cases}
\end{equation}
and $\mathcal{I}_{3}$ can be transformed into each other by swapping parties $1$ and $1'$, we find that 
\begin{equation}
a_{1';0}a_{2;1}a_{3;0}=1\,, \quad a_{i;0}a_{i+1;1}a_{i+2;0}=1\, \quad \textrm{for all } i\in(3,\dots,N-3)_{2}\,,\qquad a_{N-1;0}a_{N;1}=1\,
\end{equation}
hold true in $\mathcal{I}_{2}$.

Taking a product of all of the relations in $\mathcal{I}_{2}$ with appropriate powers gives us
\begin{equation}\label{eq:a_product_odd}
(a_{1';0}a_{2;1}a_{3;0})^{\chi_{2}} (a_{3;0}a_{4;1}a_{5;0})^{\chi_{4}}\dots (a_{N-1;0}a_{N;1})^{\chi_{N}}    =a_{1';0}^{\chi_{1}}a_{2;1}^{\chi_{2}}a_{4;1}^{\chi_{4}}\dots,a_{N;1}^{\chi_{N}}=1\,,
\end{equation}
where we use the fact that $\chi_{1}=\chi_{2}$. This relation depends only on measurements by even-numbered parties plus the party $1'$. Since the appropriate subnetworks in $\mathcal{I}_{2}$ and $\mathcal{I}_{0}$ are equivalent, this implies that Eq. \eqref{eq:a_product_odd} holds true for $\mathcal{I}_{0}$. Therefore,
\begin{equation}
P(a_{1';0}^{\chi_{1}}a_{2;1}^{\chi_{2}}a_{4;1}^{\chi_{4}}\dots,a_{N;1}^{\chi_{N}}=1)=1\,
\end{equation}
for $\mathcal{I}_{0}$. Since this probability is strictly $1$, it follows that
\begin{equation}
P(a_{1';0}^{\chi_{1}}a_{2;1}^{\chi_{2}}a_{4;1}^{\chi_{4}}\dots,a_{N;1}^{\chi_{N}}=1|a_{1;1}^{\xi_{1}}a_{3;1}^{\xi_{3}}\dots a_{N-3;1}^{\xi_{N-3}}=1)=1\,.   
\end{equation}
Combining this result with \eqref{eq:cglmp_I1} results in the following bound,
\begin{equation}
I_{N-1,N}|_{a_{1;1}^{\xi_{1}}a_{3;1}^{\xi_{3}}\dots a_{N-3;1}^{\xi_{N-3}}=1}+d(1-P(a_{1';1}^{\chi_{1}}a_{2;0}^{\chi_{2}}a_{4;0}^{\chi_{4}}\dots,a_{N;0}^{\chi_{N}}=1|a_{1;1}^{\xi_{1}}a_{3;1}^{\xi_{3}}\dots a_{N-3;1}^{\xi_{N-3}}=1)) < d-1\,,
\end{equation}
which violates monogamy \eqref{eq:monogamy_qudit}.

\textit{Part 2: odd $N$}

Since this part of the proof is very similar to the case of even $N$, we only showcase a general proof technique while skipping details explained in the previous part.

Our goal is to perform a proof by contradiction: we show that the assumption that \eqref{eqn: CGLMP_violation}, \eqref{eqn: all_stabs_qudit} and \eqref{eqn: last_two} is fulfilled by a correlation originating from $\mathcal{O}$ results in a violation of the following monogamy relation:
\begin{equation}\label{eq:bound_qudit_odd}
I_{N-1,N}|_{a_{1';0}^{\xi_{1}}a_{2;1}^{\xi_{2}}a_{4;1}^{\xi_{4}}\dots a_{N-3;1}^{\xi_{N-3}}=1}+d(1-P(a_{1;1}^{\chi_{1}}a_{3;1}^{\chi_{3}}\dots,a_{N;1}^{\chi_{N}}=1|a_{1';0}^{\xi_{1}}a_{2;1}^{\xi_{2}}a_{4;1}^{\xi_{4}}\dots a_{N-3;1}^{\xi_{N-3}}=1))\geqslant d-1\,,
\end{equation}
where $I_{N-1,N}$ is the CGLMP inequality \eqref{eq:cglmp_omega}, $\chi_{j}$ is defined in Eq. \eqref{eq:chi_def}, and $\xi_{j}=-\chi_{j}\chi_{N-3}^{-1}$. 

Looking at the subnetworks consisting of parties $i,i+1,i+2$ for $i\in(1,N-2)_{2}$ in $\mathcal{O}$ and $\mathcal{I}_{3}$ \eqref{eq:inflation_3_d}, we can conclude that, by our assumption, the relations 
\begin{equation}
I_{N-1,N}|_{a_{N-2;0}=1}<d-1\,,\qquad a_{i;0}a_{i+1;1}a_{i+2;0}=1 \quad \textrm{for } i\in(1,N-4)_{2}\,,
\end{equation}
hold true in $\mathcal{I}_{3}$. Swapping parties $1$ and $1'$ gives us the equivalent relations in $\mathcal{I}_{2}$ \eqref{eq:inflation_2_d}, which we can use to derive the following:
\begin{equation}
a_{1';0}^{\chi_{1}} a_{2;1}^{\chi_{2}}a_{4;1}^{\chi_{4}}\dots a_{N-3;1}^{\chi_{N-3}}a_{N-2;0}^{\chi_{N-3}}=1\,.
\end{equation}
We use this to derive a formula for $a_{N-2;0}$, which allows us to show that
\begin{equation}\label{eq:bound_violation_odd}
I_{N-1,N}|_{a_{1';0}^{\xi_{1}} a_{2;1}^{\xi_{2}}a_{4;1}^{\xi_{4}}\dots a_{N-3;1}^{\xi_{N-3}}=1}<d-1\,
\end{equation}
holds true in $\mathcal{I}_{2}$. Next, analysing the appropriate subnetworks of $\mathcal{I}_{2}$ and $\mathcal{I}_{0}$ \eqref{eq:inflation_0_d}, one can conclude that the above relation also holds in $\mathcal{I}_{0}$.

On the other hand, from our assumptions and from the equivalence of the appropriate subnetworks, we can conclude that 
\begin{equation}
a_{i;0}a_{i+1;1}a_{i+2;0}=1 \quad \textrm{for all } i\in(0,N-3)_{2}\,,\qquad a_{N-1;0}a_{N;1}=1\,
\end{equation}
holds true in $\mathcal{I}_{1}$ \eqref{eq:inflation_1_d}. Taking a product of these expressions with appropriate powers gives us
\begin{equation}
a_{1;1}^{\chi_{1}}a_{3;1}^{\chi_{3}}\dots,a_{N;1}^{\chi_{N}}=1\,.
\end{equation}
This can be shown to hold true in $\mathcal{I}_{0}$ by the equivalence of subnetworks containing only odd-numbered, non-primed parties in inflations $\mathcal{I}_{1}$ and $\mathcal{I}_{0}$. We therefore have that
\begin{equation}
P(a_{1;1}^{\chi_{1}}a_{3;1}^{\chi_{3}}\dots,a_{N;1}^{\chi_{N}}=1|a_{1';0}^{\xi_{1}}a_{2;1}^{\xi_{2}}a_{4;1}^{\xi_{4}}\dots a_{N-3;1}^{\xi_{N-3}}=1)=1\,.
\end{equation}

Combining the above with \eqref{eq:bound_violation_odd}, we get 
\begin{equation}
I_{N-1,N}|_{a_{1';0}^{\xi_{1}}a_{2;1}^{\xi_{2}}a_{4;1}^{\xi_{4}}\dots a_{N-3;1}^{\xi_{N-3}}=1}+d(1-P(a_{1;1}^{\chi_{1}}a_{3;1}^{\chi_{3}}\dots,a_{N;1}^{\chi_{N}}=1|a_{1';0}^{\xi_{1}}a_{2;1}^{\xi_{2}}a_{4;1}^{\xi_{4}}\dots a_{N-3;1}^{\xi_{N-3}}=1))< d-1\,,
\end{equation}
which violates Ineq. \eqref{eq:bound_qudit_odd}.
\end{proof}

\section{The GHZ state is LOSR-GMNL for any local dimension}\label{app:ghz}
In this section, we generalize the results from \cite{PhysRevA.104.052207} to show that a generalized Greenberger-Horne-Zeilinger (GHZ) \cite{Greenberger1989} state
\begin{equation}
\ket{GHZ_{N}^{d}}=\sum_{j=0}^{d-1} \ket{j}^{\otimes N}
\end{equation}
is LOSR-$\GMNL$ for any $d$. Crucially, this state is stabilized by the following operators
\begin{align}\label{eq:ghz_stab}
\begin{split}
g_{1} &= X_{1}X_{2}\dots X_{N}\,,\\
g_{j} &= Z_{j-1}Z_{j}^{-1} \quad \textrm{for } j\in(2,N)_{1}\,,
\end{split}
\end{align}
where $X$ and $Z$ are defined in Eq. \eqref{eq:def_xz}.

As was the case in Appendix \ref{app:C_N^d}, here we also make use of the CGLMP inequality \cite{PhysRevLett.88.040404} defined in \eqref{eq:cglmp_omega}. For more details on this inequality, see Appendix \ref{app:C_N^d}. We also use the same assumption and notation for inflations as established in Appendix \ref{app:C_N}.

\begin{theorem}
A GHZ state $\ket{GHZ_{N}^{d}}$ is LOSR-$\GMNL$ for any $d$ and $N\geqslant 3$.
\end{theorem}

\begin{proof}
We prove the theorem by contradiction. To this end, let us take the observables as follows: 
\begin{equation}\label{eqn: gen_GHZ_meas}
A_{1;k}=F^{\dagger}U_{\alpha_{k}}X^{\dagger}U_{\alpha_{k}}^{\dagger}F \,,\qquad A_{2;k}=FU_{\beta_{k}}XU_{\beta_{k}}^{\dagger}F^{\dagger}\,, \qquad A_{2;0}=Z\,, \qquad A_{i;j}= Z^{1-j}X^{j}\quad\,,
\end{equation}
where $k\in\{1,2\}$, $j\in\{0,1\}$, $i\in (3,N)_{1}$, $F$ is the discrete Fourier transform matrix and $U_{\chi}=\sum_{j}\omega^{j\chi}\ket{j}\!\bra{j}$ for any $\chi$, with $\alpha_{1}=0$, $\alpha_{2}=1/2$, $\beta_{1}=1/4$, $\beta_{2}=-1/4$. Note that $A_{1;1}=F^{\dagger}X^{\dagger}F=Z$. Let us denote by $\ket{\phi}_{3,\dots, N}$ the state on parties $3,\dots,N$ after the measurement that results in $a_{3;1}a_{4;1}\dots a_{N;1}=1$. From \eqref{eq:ghz_stab}, it follows that the global state $\ket{\psi}$ after the measurement fulfills
\begin{equation}
\ket{\psi} = \eta \ket{\phi}_{3,\dots,N}\bra{\phi} \!\ket{GHZ_{N}^{d}}=\eta \ket{\phi}_{3,\dots,N}\bra{\phi} g_{1}\ket{GHZ_{N}^{d}}=\eta X_{1}X_{2}\ket{\phi}_{3,\dots,N}\bra{\phi} \!\ket{GHZ_{N}^{d}}=X_{1}X_{2}\ket{\psi}\,,
\end{equation}
where $\eta$ is a normalization constant. Similarly, using $g_2$, one can show that $\ket{\psi}=Z_{1}Z_{2}^{-1}\ket{\psi}$, therefore
\begin{equation}
\ket{\psi}=\ket{\phi^{+}}_{1,2}\ket{\phi}_{3,\dots,N}\,,
\end{equation}
where $\ket{\phi^{+}}_{1,2}=1/\sqrt{d}\sum_{j=0}^{d-1}\ket{jj}$. We can then use the fact that $F\otimes F^{\dagger}\ket{\phi^{+}}=\ket{\phi^{+}}$ to conclude that, using observables $A_{1;1},A_{1;2},A_{2;1},A_{2;2}$ from Eq. \eqref{eqn: gen_GHZ_meas}, we get
\begin{equation}\label{eq:I_21}
I_{1,2}|_{a_{3;1}a_{4;1}\dots a_{N;1}=1}<d-1\,.
\end{equation}
Moreover, from \eqref{eq:ghz_stab} we have 
\begin{align}\label{eqn: ghz_stab_ob}
    \langle A_{1;1}A_{2;0}^{-1} \rangle =1\,, \quad \langle A_{i;0} A_{i+1;0}^{-1} \rangle=1 \quad \textrm{for } i\in(2,N-1)_{1}\,.
\end{align}
Note that measuring the observable $Z^{-1}$ is equivalent to measuring $Z$ but assigning to each outcome the multiplicative inverse of the corresponding outcome of $Z$, so from \eqref{eqn: ghz_stab_ob} we have  
\begin{align}
a_{1;1}a_{2;0}^{-1}=1\,, \quad a_{i;0}a_{i+1;0}^{-1} =1 \quad \textrm{for } i\in(2,N-1)_{1}\,, 
\end{align}
and therefore 
\begin{align}\label{eqn: same_GHZ}
a_{1;1}=a_{2;0}\,, \quad a_{i;0} =a_{i+1;0} \quad \textrm{for } i\in(2,N-1)_{1}\,. 
\end{align}
We will now show that the correlations \eqref{eq:I_21} and \eqref{eqn: same_GHZ} cannot be produced simultaneously by a network $\mathcal{O}$ consisting of $N$ parties connected by $(N-1)$-partite sources of non-signalling correlations and an $N$-partite source of shared randomness.

Let us consider the following inflation of $\mathcal{O}$
\begin{equation}
\mathcal{I}_{0}:\quad \tau_{i}= (1,N)_{1}\setminus\{i\}\,.
\end{equation}

First, we use the results of \cite{Augusiak_2014} to establish the following monogamy relation, 
\begin{align}\label{eq:ghz_monogamy}
&I_{1,2}|_{a_{3;1}a_{4;1}\dots a_{N;1}=1}+d(1-P(a_{1;1}=a_{N';0}|a_{3;1}a_{4;1}\dots a_{N;1}=1))\geqslant d-1\,,
\end{align}
which holds in $\mathcal{I}_0$. The CGLMP term does not depend on any primed parties, and, since $\mathcal{I}_{0}$ consists of two disconnected copies of $\mathcal{O}$, we have that its value is exactly the same in $\mathcal{I}_{0}$ as in $\mathcal{O}$. If we assume then that \eqref{eq:I_21} can be achieved in $\mathcal{O}$, it must also be achieved in $\mathcal{I}_{0}$. Consider now the intermediary inflation
\begin{equation}
\mathcal{I}_{1}:\quad \tau_{i}= (i+1,N)_{1}\,.
\end{equation}
Notice that any subnetwork consisting of two parties $j, j+1$ for all $j\in(1,N-1)_{1}$ has the same structure in $\mathcal{O}$ and in $\mathcal{I}_{1}$. Hence, if the correlations \eqref{eqn: same_GHZ} hold in $\mathcal{O}$, they will also hold in $\mathcal{I}_1$. What's more, one can combine all of the equalities in \eqref{eqn: same_GHZ} to find 
\begin{equation}
a_{1;1}=a_{N;0}\,.
\end{equation}
In $\mathcal{I}_{1}$, parties $1$ and $N$ do not share any common source, which is also true for parties $1$ and $N'$ in $\mathcal{I}_{0}$. Therefore, in $\mathcal{I}_{0}$ we have
\begin{equation}
a_{1;1}=a_{N';0}\,,
\end{equation}
so, trivially,
\begin{align}
P(a_{1;1}=a_{N';0}|a_{3;1}a_{4;1}\dots a_{N;1}=1)=P(a_{1;1}=a_{N';0})=1\,.
\end{align}
Combined with \eqref{eq:I_21}, this implies that
\begin{align}
I_{1,2}|_{a_{3;1}a_{4;1}\dots a_{N;1}=1}+d(1-P(a_{1;1}=a_{N';0}|a_{3;1}a_{4;1}\dots a_{N;1}=1))< d-1\,.
\end{align}
This contradicts the monogamy relation \eqref{eq:ghz_monogamy} and proves that $\ket{GHZ_{N}^{d}}$ is LOSR-$\GMNL$.   
\end{proof}

\section{LONC-GMNL setting}\label{app:c-gmnl}
\subsection{$\ket{{\rm GHZ}_n}$ is maximally LONC-GMNL on a directed path}\label{sectionGHZonaline}
We prove Theorem~\ref{GHZline} which appears in the main text.
\begin{theorem}[Full statement of Theorem~\ref{GHZline}]\label{GHZlinefull}
In the LONC-$\GMNL$ model, the following inequality holds if we allow $t<n-1$ rounds of synchronous communication along an oriented path, where the $n$ parties are named ${A^{(1)}\rightarrow \dots\rightarrow A^{(n)}}$.
\begin{align}
{\langle A_1^{(1)} \cdots A_1^{(n-1)}  A_{1}^{(n)} \rangle}+{\langle A_1^{(1)}\cdots A_1^{(n-1)}  A_{2}^{(n)} \rangle}+{\langle A_0^{(1)} A_{1}^{(n)} \rangle}-{\langle A_0^{(1)} A_{2}^{(n)} \rangle} +2{\langle A_0^{(1)} A_0^{(n)} \rangle} \le 4 \,.\label{eqghzline}
\end{align}
It is violated by the quantum state $\ket{{\rm GHZ}_n}$ when measuring the observables $Z$ for ${A_0^{(1)},A_0^{(n)}}$; $X$ for  $A_1^{(1)},\dots,A_1^{(n-1)}$; $\frac{X+Z}{\sqrt{2}}$ for $A_1^{(n)}$; and $\frac{X-Z}{\sqrt{2}}$ for $A_2^{(n)}$.
\end{theorem}
\begin{proof}
The quantum violation is straightforward. We prove the inequality by invoking the monogamy of bipartite CHSH non-signaling correlations~\cite{Augusiak_2014}. We first remark that $t<n-1$ implies that no causal resources, except for the starting classical shared randomness, can be shared by the extremal parties $A^{(1)}$ and $A^{(n)}$ in an oriented path. An external party $\tilde{A}^{(n)}$ having access only to the shared randomness can thus achieve the same correlation $2\langle A_0^{(1)} A_0^{(n)} \rangle=2\langle A_0^{(1)} \tilde{A}_0^{(n)} \rangle$. Then, we observe that the first four terms of the left-hand side of Eq.~\eqref{eqghzline} are reducible to the CHSH inequality by regrouping the parties into two groups: $(A^{(1)}\dots A^{(n-1)})$ and $A^{(n)}$. We note that the oriented path structure prevents signalling between those two groups. More precisely, the inputs $a\in\{0,1\}$ given to $A^{(1)}$ starts too far to reach $A^{(n)}$ in $t<n-1$ steps and $a$ thus cannot effect the marginal probability distribution $P(A_b^{(n)}|a)=P(A_b^{(n)})$; and the other input, let it be $b\in\{0,1,2\}$, given to $A^{(n)}$ cannot effect the joint marginal probability distribution $P(A_a^{(1)}\dots A_1^{(n-1)}|b)=P(A_a^{(1)}\dots A_1^{(n-1)})$ because the signalling is only left to right. We conclude by applying the monogamy inequality for bipartite non-signaling correlations that are given in~\cite[Equation 2]{Augusiak_2014}.
\end{proof}

\subsection{Cluster states are (only) LONC-$\GMNL_2$ on a directed path}
Here we show a way of generating a cluster state in 2 rounds of communication on an $N$-partite line network in which communication is allowed only in one direction. For simplicity, we take that the first party can send information to the second, the second to the third, etc. Each party $i\neq 1,N$ has access to two qubits, which we denote $(i,1)$ and $(i,2)$, each of them initially in a state $\ket{+}$. As for the first and last party, the first one starts with three qubits in a state $\ket{+}$, which we denote by $(1,1)$, $(1,2)$ and $(0,1)$, while the party $N$ starts with no qubits. We thus start with a graph state corresponding to an edgeless graph with $2N-1$ vertices. In what follows, we give an algorithm to generate a cluster state in one-way line networks; however, we will describe the operations in terms of graph transformations. Afterward, we comment on how these transformations relate to the operations performed on qubits.

\begin{enumerate}
    \item Each party $i$ connects vertices $(i,1)$ and $(i,2)$ with an edge. The first party additionally connects $(0,1)$ and $(1,1)$.
\item Each party $i\neq N$ sends vertex (qubit) $(i,2)$ to the party $i+1$.
\item Each party $i\neq 1,N$ connects $(i,1)$ and $(i-1,2)$ with an edge.
\item We perform local complementation on vertices $(i,2)$ which in turn connects vertices $(i,1)$ and $(i+1,1)$ for all $i\neq 0$.
\item Each party $i\neq 1,N$ disconnects vertices $(i,1)$ and $(i-1,2)$.
\item Party $i$ sends vertex $(i,1)$ to the party $i+1$. 
\item Each party $i\neq 1$ disconnects vertices $(i-1,1)$ and $(i-1,2)$. The remaining graph is a linear graph of qubits (i,1) for all $i\in\{0,\dots,N-1 \}$ and isolated vertices $(i,2)$ for all $i\in \{1,\dots, N\}$.
\end{enumerate}

The procedure is illustrated for $N=4$ in Fig.~\ref{fig:qconstruction}.
The operation of connecting and disconnecting vertices corresponds to the action of $C_{Z}=\operatorname{diag}(1,1,1,-1)$ on the corresponding qubits. In the above algorithm, this operation is always performed locally, i.e. on qubits that a given party has access to. As for the local complementation, it is known that this operation corresponds to action with local unitaries on individual qubits (see Appendix \ref{app:graph_states} for more details).

~

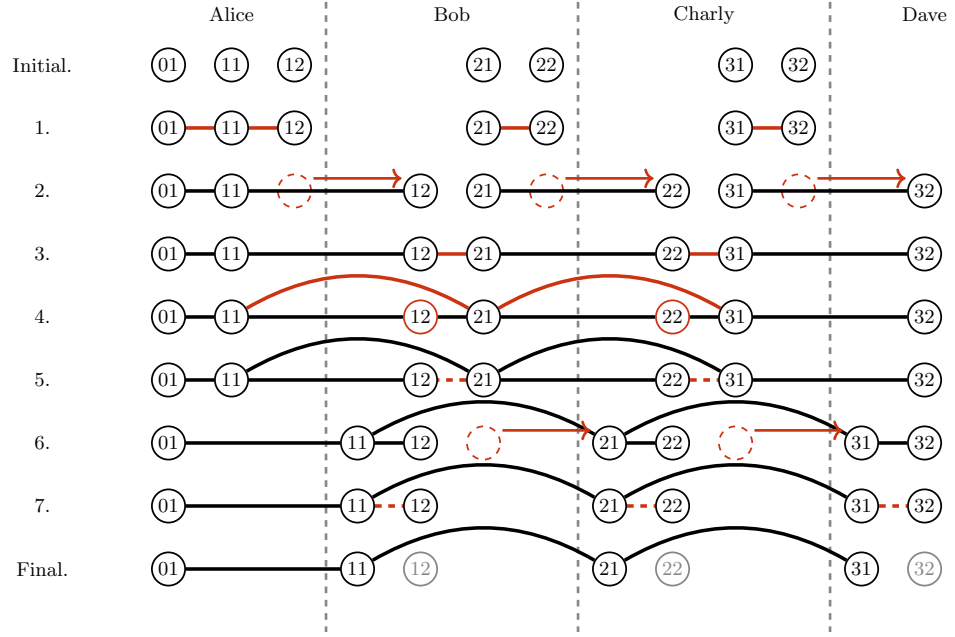
\begin{figure}[hbt]
    \centering
   \resizebox{0.75\textwidth}{!}{\definecolor{myred}{RGB}{204,51,17}
\definecolor{myblue}{RGB}{0,119,187}
\definecolor{mygrey}{RGB}{135,135,135}
\definecolor{myteal}{RGB}{80,175,148}
\definecolor{myorange}{RGB}{238,119,51}
\definecolor{mymagenta}{RGB}{238,51,110}

\tikzset{player/.style={circle,draw=black,inner sep=0, minimum size=15pt,fill=white,thick
 }}

 \tikzset{playerremoved/.style={circle,draw=myred,dashed,inner sep=0, minimum size=15pt,fill=white,thick
 }}

\tikzset{ligne/.style={very thick
}}

\centering
\begin{tikzpicture}[scale=1]
\centering

\usetikzlibrary{decorations.pathreplacing}
\tikzset{ligne/.style={ultra thick,myred
}}
\tikzset{ligne2/.style={ultra thick,black
}}

\begin{scope}[shift={(0,0)}]
\node (step1) at (-2,0) {Initial.};
\node[player] (01) at (0,0) {$01$};
\node[player] (11) at (1,0) {$11$}; 
\node[player] (12) at (2,0) {$12$}; 

\node[player] (21) at (5,0) {$21$};
\node[player] (22) at (6,0) {$22$}; 
\node[player] (31) at (9,0) {$31$}; 
\node[player] (32) at (10,0) {$32$};
\end{scope}

\begin{scope}[shift={(0,-1)}]
\node (step2) at (-2,0) {1.};
\node[player] (01) at (0,0) {$01$};
\node[player] (11) at (1,0) {$11$}; 
\node[player] (12) at (2,0) {$12$}; 

\node[player] (21) at (5,0) {$21$};
\node[player] (22) at (6,0) {$22$}; 
\node[player] (31) at (9,0) {$31$}; 
\node[player] (32) at (10,0) {$32$};

\draw[ligne] (01) -- (11);
\draw[ligne] (11) -- (12);
\draw[ligne] (21) -- (22);
\draw[ligne] (31) -- (32);
\end{scope}

\begin{scope}[shift={(0,-2)}]
\node (step2) at (-2,0) {2.};
\node[player] (01) at (0,0) {$01$};
\node[player] (11) at (1,0) {$11$}; 
\node[playerremoved] (12o) at (2,0) {}; 
\node[player] (12) at (4,0) {$12$}; 

\node[player] (21) at (5,0) {$21$};
\node[playerremoved] (21r) at (6,0) {};
\node[player] (22) at (8,0) {$22$}; 
\node[player] (31) at (9,0) {$31$}; 
\node[playerremoved] (31r) at (10,0) {}; 
\node[player] (32) at (12,0) {$32$};

\draw[ligne2] (01) -- (11);
\draw[ligne2] (11) -- (12);
\draw[ligne2] (21) -- (22);
\draw[ligne2] (31) -- (32);
\draw[ligne,very thick,->] (2.3,0.2) -- (3.7,0.2);
\draw[ligne,very thick,->] (6.3,0.2) -- (7.7,0.2);
\draw[ligne,very thick,->] (10.3,0.2) -- (11.7,0.2);
\end{scope}

\begin{scope}[shift={(0,-3)}]
\node (step2) at (-2,0) {3.};
\node[player] (01) at (0,0) {$01$};
\node[player] (11) at (1,0) {$11$}; 
\node[player] (12) at (4,0) {$12$}; 

\node[player] (21) at (5,0) {$21$};
\node[player] (22) at (8,0) {$22$}; 
\node[player] (31) at (9,0) {$31$}; 
\node[player] (32) at (12,0) {$32$};

\draw[ligne2] (01) -- (11);
\draw[ligne2] (11) -- (12);
\draw[ligne2] (21) -- (22);
\draw[ligne2] (31) -- (32);
\draw[ligne] (21) -- (12);
\draw[ligne] (22) -- (31);
\end{scope}

\begin{scope}[shift={(0,-4)}]
\node (step2) at (-2,0) {4.};
\node[player] (01) at (0,0) {$01$};
\node[player] (11) at (1,0) {$11$}; 
\node[player,draw=myred] (12) at (4,0) {$12$}; 

\node[player] (21) at (5,0) {$21$};
\node[player,draw=myred] (22) at (8,0) {$22$}; 
\node[player] (31) at (9,0) {$31$}; 
\node[player] (32) at (12,0) {$32$};

\draw[ligne2] (01) -- (11);
\draw[ligne2] (11) -- (12);
\draw[ligne2] (21) -- (22);
\draw[ligne2] (31) -- (32);
\draw[ligne2] (21) -- (12);
\draw[ligne2] (22) -- (31);
\draw[ligne] (11) to[bend left] (21);
\draw[ligne] (21) to[bend left] (31);
\end{scope}

\begin{scope}[shift={(0,-5)}]
\node (step2) at (-2,0) {5.};
\node[player] (01) at (0,0) {$01$};
\node[player] (11) at (1,0) {$11$}; 
\node[player] (12) at (4,0) {$12$}; 

\node[player] (21) at (5,0) {$21$};
\node[player] (22) at (8,0) {$22$}; 
\node[player] (31) at (9,0) {$31$}; 
\node[player] (32) at (12,0) {$32$};

\draw[ligne2] (01) -- (11);
\draw[ligne2] (11) -- (12);
\draw[ligne2] (21) -- (22);
\draw[ligne2] (31) -- (32);
\draw[ligne,dashed] (21) -- (12);
\draw[ligne,dashed] (22) -- (31);
\draw[ligne2] (11) to[bend left] (21);
\draw[ligne2] (21) to[bend left] (31);
\end{scope}

\begin{scope}[shift={(0,-6)}]
\node (step2) at (-2,0) {6.};
\node[player] (01) at (0,0) {$01$};
\node[player] (11) at (3,0) {$11$}; 
\node[player] (12) at (4,0) {$12$}; 
\node[playerremoved] (21r) at (5,0) {$$};
\node[player] (21) at (7,0) {$21$};
\node[player] (22) at (8,0) {$22$}; 
\node[playerremoved] (31r) at (9,0) {};
\node[player] (31) at (11,0) {$31$}; 
\node[player] (32) at (12,0) {$32$};

\draw[ligne2] (01) -- (11);
\draw[ligne2] (11) -- (12);
\draw[ligne2] (21) -- (22);
\draw[ligne2] (31) -- (32);
\draw[ligne2] (11) to[bend left] (21);
\draw[ligne2] (21) to[bend left] (31);
\draw[ligne,very thick,->] (5.3,0.2) -- (6.7,0.2);
\draw[ligne,very thick,->] (9.3,0.2) -- (10.7,0.2);
\end{scope}

\begin{scope}[shift={(0,-7)}]
\node (step2) at (-2,0) {7.};
\node[player] (01) at (0,0) {$01$};
\node[player] (11) at (3,0) {$11$}; 
\node[player] (12) at (4,0) {$12$}; 

\node[player] (21) at (7,0) {$21$};
\node[player] (22) at (8,0) {$22$}; 
\node[player] (31) at (11,0) {$31$}; 
\node[player] (32) at (12,0) {$32$};

\draw[ligne2] (01) -- (11);
\draw[ligne,dashed] (11) -- (12);
\draw[ligne,dashed] (21) -- (22);
\draw[ligne,dashed] (31) -- (32);
\draw[ligne2] (11) to[bend left] (21);
\draw[ligne2] (21) to[bend left] (31);
\end{scope}

\begin{scope}[shift={(0,-8)}]
\node (step2) at (-2,0) {Final.};
\node[player] (01) at (0,0) {$01$};
\node[player] (11) at (3,0) {$11$}; 
\node[player,draw=mygrey] (12) at (4,0) {\textcolor{mygrey}{$12$}}; 

\node[player] (21) at (7,0) {$21$};
\node[player,draw=mygrey] (22) at (8,0) {\textcolor{mygrey}{$22$}}; 
\node[player] (31) at (11,0) {$31$}; 
\node[player,draw=mygrey] (32) at (12,0) {\textcolor{mygrey}{$32$}};

\draw[ligne2] (01) -- (11);
\draw[ligne2] (11) to[bend left] (21);
\draw[ligne2] (21) to[bend left] (31);
\end{scope}

\begin{scope}[on background layer]

\node (A) at (1,0.8) { {Alice}};
\node (B) at (4.5,0.8) { {Bob}};
\node (C) at (8.5,0.8) { {Charly}};
\node (D) at (12,0.8) { {Dave}};

\draw[very thick,mygrey,dashed] (2.5,1) -- (2.5,-9);
\draw[very thick,mygrey,dashed] (6.5,1) -- (6.5,-9);
\draw[very thick,mygrey,dashed] (10.5,1) -- (10.5,-9);
\end{scope}
\end{tikzpicture}}
    \caption{An illustration for $N=4$ of our quantum algorithm to prepare a cluster state with local operations and two rounds of synchronous communication on a directed path.
   }\label{fig:qconstruction}
\end{figure}

\end{document}